\documentclass[12pt]{article}

\usepackage{graphicx}
\usepackage{authblk}

\usepackage{tikz}\usetikzlibrary{automata, positioning, arrows}

\tikzset{->, >=stealth', initial text=$ $}

\usepackage{amsmath,amsthm,amssymb,xspace,xcolor,stmaryrd,hyperref,fullpage,mathdots,MnSymbol,wasysym}
\hypersetup{colorlinks=true, linkcolor=red,citecolor=blue,urlcolor=blue}

\newcommand{\moritz}[1]{}

\newtheorem{theorem}{Theorem}
\newtheorem{lemma}[theorem]{Lemma}
\newtheorem{proposition}[theorem]{Proposition}
\newtheorem{corollary}[theorem]{Corollary}

\theoremstyle{definition}
\newtheorem{definition}[theorem]{Definition}
\newtheorem{example}[theorem]{Example}
\newtheorem{remark}[theorem]{Remark}

\newcommand{\PTIME}{\mathsf{PTIME} }

\newcommand{\W}{\mathsf{W} }
\newcommand{\FPT}{\mathsf{FPT} }
\newcommand{\PSPACE}{\mathsf{PSPACE} }
\newcommand{\LTL}{\mathsf{LTL} }
\newcommand{\MSO}{\mathsf{MSO} }
\newcommand{\FO}{\mathsf{FO} }
\newcommand{\MTL}{\mathsf{MTL} }
\newcommand{\TPTL}{\mathsf{TPTL} }
\newcommand{\TCTL}{\mathsf{TCTL} }
\newcommand{\MITL}{\mathsf{MITL} }
\newcommand{\SWA}{\mathsf{SWA} }

\newcommand{\N}{\mathbb{N}}

\newcommand{\R}{\mathbb{R}}

\newcommand{\A}{\mathbb{A}}
\newcommand{\B}{\mathbb{B}}

\newcommand{\G}{\mathcal{G}}
\newcommand{\Act}{\mathcal{A}}

\newcommand{\TS}{\mathit{TS}}

\renewcommand{\cal}{\mathcal }

\renewcommand{\le}{\leqslant}
\renewcommand{\ge}{\geqslant}

\begin{document}

\title{ Model-checking in the Foundations of  Algorithmic Law\\ and the Case of Regulation 561}

\author[1]{Moritz M\"uller}
\author[2]{Joost J. Joosten}
\affil[1]{University of Passau}
\affil[2]{University of Barcelona}
\maketitle

\begin{abstract}
We discuss model-checking problems as formal models of algorithmic law. Specifically, we ask for an algorithmically tractable general purpose model-checking problem that naturally models the European transport Regulation 561 (\cite{Reg561}), and discuss the reaches and limits of a version of discrete time stopwatch automata.
\end{abstract}

\section{Model-checking and algorithmic law}

%
%

The European transport Regulation 561~\cite{Reg561} concerns 
 activities of truck drivers as recorded by tachographs. A tachograph recording determines for each time unit the activity of the driver which can be \textit{driving}, \textit{resting} or \textit{doing other work}. Regulation 561 is a complex set of articles that limits driving and work time by prescribing various types of rest periods. The regulation prescribes that the time units are minutes, so a tachograph recording of 2 months determines a sequence of activities of length 87840. It is clear that the legality of
such a recording can only be judged with the help of an algorithm.  


By the application of a law to a case we mean the decision whether the case is legal according to that law or not.
By an {\em algorithmic law} we mean a law whose application to a case is executed by an algorithm. Instead of designing one algorithm per law we are interested in {\em general purpose} algorithms: these take as input both a case from a set of cases of interest, and a law from a set of laws of interest, and decide whether the given case is legal according to the given law or not.
 In order to present cases and laws of interest as inputs to an algorithm, both have to be suitably \textit{formalized}.

\subsection{Computational problems in algorithmic law}\label{sec:compprobl}

For Regulation 561,  a case is a sequence of activities and hence straightforwardly formalized as a word over the alphabet $\Sigma:=\{d,r,w\}$: e.g., the word $dddwrr\in\Sigma^6$ is the activity sequence consisting of 3 minutes driving, followed by 1 minute other work, followed by 2 minutes resting.\footnote{Some tachograph readers will work with other formats like \textit{activity-change lists}: lists of timepoints where the driver's activity changes. We do not discuss other formats in this paper. } Generally, we formalize a set of cases by a class of finite structures $\mathcal K$.\footnote{Words are straightforwardly seen as structures, see e.g.~\cite[Example~4.11]{fg}.}

In this setting, a generic formalization of a law is given by translating the law to a sentence $\varphi$ of a formal language, i.e., a logic $L$. That a particular case $K\in\mathcal K$ is legal according to the law $\varphi$ then formally means that $K\models\varphi$, i.e., $K$ satisfies $\varphi$. We arrive at what is the central computational problem of algorithmic law:

\paragraph{Model-checking} The {\em model-checking problem (for $L$ over  $\cal K$)} is a formal model for  a family of algorithmic laws where  laws are formalized by sentences of $L$ and  cases are formalized by structures in $\cal K$. 

\medskip
\fbox{\parbox{13cm}{
\textsc{MC($\cal K,L$)}\\
\begin{tabular}{ll}
{\em Input:}& $K\in\cal K$ and $\varphi\in L$.\\
{\em Problem:}& $K\models \varphi$ ?
\end{tabular}
}}
\medskip

A  {\em model-checker (for $L$ over $\cal K$)} is an algorithm deciding $\textsc{MC($\cal K,L$)}$. This is a
 general purpose algorithm as asked for above. \\

\noindent 
We consider two more computational problems associated to algorithmic law. 

\paragraph{Consistency-checking} 

A  minimal requirement for law design is that it should be possible to comply with the law (cf.\ \cite{Errezil:2019:Homologation} for a problematic case). 
For laws governing activity sequences {\em consistency} means that there should be at least one such sequence that is legal according to the law. 
A related question of interest is whether a certain type of behaviour can be  legal. This is tantamount to ask whether the artificial law augmented by demanding the type of behaviour is consistent. 

This  is formally modeled by the {\em consistency problem (for $L$ over $\cal K$)}:

\medskip
\fbox{\parbox{13cm}{
\textsc{Con($\cal K,L$)}\\
\begin{tabular}{ll}
{\em Input:}& $\varphi\in L$.\\
{\em Problem:}& does there exist some $K\in\cal K$ such that $K\models \varphi$ ?
\end{tabular}
}}

\paragraph{Scheduling} Assume a truck driver has to schedule next week's driving, working and resting and is interested to drive as long as possible. A week has 10080 minutes, so the driver faces the computational optimization problem to compute a length 10080 extension of the word given by the current tachograph recording that is legal according to Regulation~561 and that maximizes driving time. 

Consider laws governing activity sequences, that is, $\cal K$ is the (set of structures corresponding to the) set of finite words $\Sigma^*$ over some alphabet $\Sigma$.  For a word $w=a_0\cdots a_{n-1}\in\Sigma^n$ (the $a_i$ are letters that represent the corresponding activities) and a letter $a\in\Sigma$, let $\#_a(w)$ denote the number of times the letter $a$ appears in $w$, i.e., $$\#_a(w):=|\{i<n\mid a_i=a\}|.$$ 
 The {\em scheduling problem (for $L$ over $\cal K=\Sigma^*$)} is: 
 
\medskip
\fbox{\parbox{13cm}{
\textsc{Scheduling($\cal K,L$)}\\
\begin{tabular}{ll}
{\em Input:}& $\varphi\in L$, $w\in\Sigma^*$, $a\in\Sigma$  and $n\in\N$.\\
{\em Problem:}& if there is no $v\in \Sigma^n$ such that $wv\models\varphi$, then output ``illegal''; \\
&otherwise output some $\overline v\in\Sigma^n$ such that \\[1ex]
&\qquad $\#_a(w\overline v)=\max\big\{\#_a(wv)\mid v\in\Sigma^n, wv\models\varphi\big\}$.
\end{tabular}
}}

\subsection{Model-checking as a formal model}\label{sec:formalmodel}

There is a vast amount of research concerning  model-checking problems $\textsc{MC}(\cal K,L)$. The two main interpretational perspectives stem from {\em database theory} and from {\em system verification}. In database theory~\cite{sss}, $\cal K$ is viewed as a set of databases, and $L$ a set of Boolean queries.  In system verification~\cite{baier}, $\cal K$ is as a set of transition systems or certain automata that formalize concurrent systems or parallel programs, and~$L$ formalizes correctness specifications of the system, that is, properties all executions of the system should have.
We add a third interpretational perspective on model-checking problems as formal models for families of algorithmic laws. We highlight three conflicting requirements on such a formal model.

\paragraph{Tractability requirement} The first and foremost constraint for a model $\textsc{MC}(\cal K,L)$ of a family of algorithmic laws is its computational complexity. For the existence of a practically useful general purpose model-checker the problem $\textsc{MC}(\cal K,L)$ should be {\em tractable}. We argue that the notion of tractability here cannot just mean $\PTIME$, a more fine-grained complexity analysis  of $\textsc{MC}(\cal K,L)$ is required.

Classical computational complexity theory tells us that already extremely simple pairs $(\cal K,L)$ have intractable model-checking problems. An important example from database theory is that $\textsc{MC}(\cal K,L)$ is $\mathsf{NP}$-complete for $L$ the set of conjunctive queries
and $\cal K$ the set of graphs (or the single binary word 01)~\cite{chandramerlin}.  
An important example \cite{pnueli}  from system verification is that $\textsc{MC}(\cal K,L)$ is $\PSPACE$-complete for $L$ equal to linear time temporal logic $\LTL$ and  $\cal K$ the class of finite automata~\cite{SistlaClarke}.

However, this $\PSPACE$-completeness result is largely irrelevant because 
the model-checking problem is {\em fixed-parameter tractable (fpt)}, that is, it is decidable in time $f(k)\cdot n^{O(1)}$ for some computable function $f:\N\to\N$ where $n$ is the total input size and $k:=\|\varphi\|$ the size of (a reasonable binary encoding of) the input $\LTL$ formula~$\varphi$. 
In fact, we have {\em parameter dependence} $f(k)\le 2^{O(k)}$.
Informally speaking, we are mainly interested in inputs with $k\ll n$, so this can be considered tractable. In other words, the computational hardness relies on uninteresting inputs with relatively large~$k$.
In contrast, model-checking conjunctive queries over graphs is likely not fixed-parameter tractable: this is equivalent to $\FPT\neq\W[1]$, the central hardness hypothesis of parameterized complexity theory.\footnote{Grohe \cite{grohecsp} (refined in \cite{chmu1,chmu2}) gives a quite complete understanding of which conjunctive queries are tractable.} 

The focus on inputs with $k\ll n$ is common in model-checking and it is an often repeated point that a reasonable complexity analysis must take this asymmetry of the input  into account; \cite{papayanna} is an early reference addressing both perspectives from database theory and  system verification. The theoretical framework for such a fine-grained complexity analysis is parameterized complexity theory \cite{fg,df,df2} whose central tractability notion is fixed-parameter tractability.\footnote{This paper does not require familiarity with parameterized complexity theory. Only Section~\ref{sec:lower} requires some results of this theory and will recall what is needed.} 

To sum up, judging he tractability of $\textsc{MC}(\cal K,L)$ should be based on a fine-grained complexity analysis 
that measures the computational complexity with respect to various input {\em aspects} $n,k,\ldots$.\footnote{Formally, an {\em aspect} could be defined as a {\em parameterization}, possibly viewed as a {\em size measure} as in \cite[p.418f]{fg}. However, we don't need a definition and use the term informally.} 
The quality of the model $\textsc{MC}(\cal K,L)$ depends on the ``right'' identification of relevant aspects in its complexity analysis.


\paragraph{Expressivity requirement} Recall that we ask for a {\em general purpose} model-checker that solves a model-checking problem $\textsc{MC}(\cal K,L)$ modeling a family  of algorithmic laws instead of single-purpose model-checkers deciding $\textsc{MC}(\cal K,\{\varphi\})$, one per algorithmic law $\varphi$. From a theoretical perspective we expect insight on which laws can possibly be algorithmic. 

From a practical perspective, this avoids the costly production of many algorithms, their updates following law reforms and their validation for legal use. It is thus desirable to find tractable $\textsc{MC}(\cal K,L)$ for as rich as possible classes $\cal K$ and $L$. In particular, for laws governing sequences of activities (i.e., $\cal K=\Sigma^*$) we ask for an as expressive as possible logic $L$. Of course, this is in tension with the tractability requirement.

\paragraph{Naturality requirement} From an algorithmic perspective it is not only the expressivity of  $L$ that matters, but also its {\em succinctness}. 
Typically, model-checking complexity grows fast with the size of the sentence $\varphi$ formalizing the law, so logics allowing for shorter formalizations are preferable. 
E.g., it is well-known that the expressive power of  $\LTL$ is not increased when adding past modalities but their use can lead to exponentially shorter sentences.
Crucially, the complexity of model-checking (over finite automata) is not substantially increased. Moving to a more succinct logic is not necessarily an improvement. E.g. further adding a now-modality again increases succinctness exponentially but apparently also the model-checking complexity~\cite{forgettable}. 

Furthermore, it is one thing to model a law application by a model-checking instance $(K,\varphi)$ any old how and another to do so by somehow typical members of $\cal K$ and $ L$. E.g., in case the formalization of actual laws uses only special artificial members of $\cal K$ ({\em semantic overkill}) or $ L$ ({\em syntactic overkill}), one would want to trade the richness of $\cal K$ and $ L$ for a faster model-checker.

Very long or contrived formalizations of laws are also prohibitive for legal practice which requires the law to be readable and understandable by humans. This is vital also for the validation of their formalization, i.e., their translation from the typically ambiguous natural language into a formal language able to be algorithmically processed. Without attempting a definition of this vague term, we thus informally require that the (formalization given by the) model $\textsc{MC}(\cal K,L)$ must be {\em natural}.

\paragraph{Other requirements} We focus on the above three requirements but, of course, there are more whose discussion is beyond the scope of this paper. 

An important one is trust in the output of model-checkers. This is a threefold issue. First, the formalization process requires trust: laws are 
written in natural language and thereby formally not precise and ambigue; formalization
 typically leads to choices to disambiguate or even repair the written law; this calls for a collaboration of different experts. Second, the implementation process requires trust: this could call for formally verified implementations; we refer to \cite{TimePaper} for an example. Third, one needs  trust that the data given to the algorithm are correct and in the right format (we refer to \cite{Errezil:2019:Homologation} for a discussion); for example, Regulation 561 prescribes working in UTC and it is known that no tachograph actually records in UTC; theoretically,  the change from non-UTC to UTC data can have drastic effects \cite{drive}.

Furthermore, algorithmic outputs should be transparent and explainable to be used in legal practice and it is unclear what this exactly means.
  Further requirements on the model might come from ethical or political considerations - e.g., the required transparency can be in conflict with intellectual property rights and there can be more general issues concerning the involvement of the private sector in law execution. 


 \subsection{Contribution and outline}\label{sec:outline}

We focus on laws governing temporal sequences of activities, that is, laws concerning cases that can readily be formalized by words over some finite alphabet $\Sigma$, i.e., $\cal K=\Sigma^*$. This paper is about the quest for a logic $L$ such that 
$\textsc{MC}(\cal K,L)$ is a good model for such laws. To judge expressivity and naturality we use European Regulation 561 \cite{Reg561} as a test case, that is, we want $L$ to naturally formalize Regulation 561. Given the complexity of this regulation, this is an ambitious goal and we expect success to result in a model that encompasses a broad family of laws concerning sequences of activities. 
 
The imperative constraint is the tractability of $\textsc{MC}(\cal K,L)$. The next section 
surveys the relevant literature on model-checking and discusses shortcomings of known model-checkers. Thereby we build up some intuition about what the right input aspects are, i.e., those relevant to calibrate the computational complexity of $\textsc{MC}(\cal K,L)$ and to judge its tractability.

We suggest (a version of) discrete time {\em stopwatch automata} $\SWA$ as an answer to our central question, that is, we propose $\textsc{MC}(\Sigma^*,\SWA)$ as a model for algorithmic laws concerning sequences of activities.

Stopwatch automata are defined in Section~\ref{sec:stopwatch}. Our main technical contribution is the construction of a stopwatch automaton expressing Regulation 561 in Section~\ref{sec:automaton}. Sections~\ref{sec:thy} and~\ref{sec:lower} gauge the expressivity of stopwatch automata and the computational complexity of the problems mentioned in Section~\ref{sec:formalmodel}: model-checking problem, consistency-checking and scheduling. It turns out that
while stopwatch automata have  high expressive power, their model-checking complexity is relatively tame, and scales well with the aspects identified in Section~\ref{sec:survey}: we summarize our technical results in Section~\ref{sec:techsum}.

\section{Regulation 561 and various logics}\label{sec:survey}

Model-checking complexity has been investigated mainly from two interpretational perspectives: database theory and system verification. We give a brief survey guided by our central question to model Regulation 561.

\subsection{Regulation 561 and B\"uchi's theorem}\label{sec:buechi}

We recall B\"uchi's theorem and, to fix some notation, the definitions of regular languages and finite automata. 

An {\em alphabet} $\Sigma$ is a non-empty finite set of {\em letters}, $\Sigma^*=\bigcup_{n\in\N}\Sigma^n$ denotes the set of (finite) {\em words}. A word $w=a_0\cdots a_{n-1}\in\Sigma^n$ (the $a_i$ are letters) has {\em length} $|w|:=n$. A {\em (non-deterministic) finite automaton} $\B$ is given by a finite set of {\em states} $Q$, an  alphabet~$\Sigma$, sets of {\em initial} and {\em final} states $I,F\subseteq Q$, and a set $\Delta\subseteq Q\times \Sigma\times Q$ of {\em transitions}. A {\em computation of $\B$ on}  $w=a_0\cdots a_{n-1}\in\Sigma^n$ is a sequence $q_0\cdots q_n$ of states such that $(q_i,a_i,q_{i+1})\in\Delta$ for every $i<n$. The computation is {\em initial} if $q_0\in I$ and {\em accepting} if $q_n\in F$. The {\em language $L(\B)$ of $\B$} is the set of words $w\in\Sigma^*$ such that $\B$ {\em accepts} $w$, i.e., there exists an initial accepting computation of $\B$ on~$w$.
A language (i.e., subset of words over $\Sigma$) is {\em regular} if it equals $L(\B)$ for some finite automaton~$\B$. 

We refer to \cite{thomas} for a definition of $\MSO$-definable languages and a proof of:

\begin{theorem}[B\"uchi]\label{thm:buechi} A language is $\MSO$-definable if and only if it is regular.
\end{theorem}

This can be extended to infinite words and trees using various types of automata -- we refer to 
 \cite{games} for a monograph on the subject. 

The proof of B\"uchi's theorem is effective in that there is a computable function that computes for every $\MSO$-sentence $\varphi$ and automaton $\B_\varphi$  that accepts a word $w$ if and only if $w\models\varphi$. It follows that
$\textsc{MC}(\Sigma^*,\MSO)$ is fixed-parameter tractable: given an input $(w,\varphi)$, compute $\B_\varphi$ and check $\B_\varphi$ accepts $w$. This takes time\footnote{This is not true for the empty word $w$. We trust the readers common sense to interpret this and similar statements  reasonably.}
 $f(\|\varphi\|)\cdot |w|$ for some computable function $f:\N\to\N$. It also follows that $\textsc{Con}(\Sigma^*,\MSO)$ is decidable because finite automata have {\em decidable emptiness}: there is an (even linear time) algorithm that, given a finite automaton $\A$, decides whether $L(\A)=\emptyset$.

$\MSO$ is a very expressive logic. In \cite{drive} it is argued that Regulation 561 can be formalized in $\MSO$, and naturally so. Thus, in a sense $\textsc{MC}(\Sigma^*,\MSO)$ is tractable, expressive and natural, so a good answer to our central question. The starting point of this work was the question for a better model, namely improving its tractability. The problem with the runtime  $f(\|\varphi\|)\cdot |w|$ of B\"uchi's model-checker is that the parameter dependence $f(k)$ grows extremely fast: it is non-elementary in the sense that it cannot be bounded by $2^{2^{\udots^{2^{k}}}}$ for any fixed height tower of 2's. This is due to the fact that in general  the size $\|\B_\varphi\|$ of (a reasonable binary encoding of) $\B_\varphi$ is non-elementary in~$\|\varphi\|$. Under suitable hardness hypotheses this non-elementary parameter dependence cannot be avoided, not even when restricting to first-order logic  $\FO$ \cite{frickgrohe}.

This motivates the quest for fragments of $\MSO$ or less succinct variants thereof that allow a tamer parameter dependence. In system verification, $\LTL$ has been proposed: an $\LTL$ formula of size~$k$ can be translated to a size $2^{O(k)}$  B\"uchi automaton \cite{vw} or a size $O(k)$ alternating automaton \cite{vardi}. The model-checking problem
asks given a system modeled by a finite automaton $\A$ whether all (finite or infinite) words accepted by the automaton satisfy the given $\LTL$-sentence $\varphi$. The model-checker decides emptiness of a suitable product automaton accepting $L(\A)\cap L(\B_{\neg \varphi})$ and takes time $2^{O(\|\varphi\|)}\cdot \|\A\|$. This is the dominant approach to model-checking in system verification. 

\cite{drive} formalizes part of Regulation 561 in  $\LTL$. In part, these formalizations rely on Kamp's theorem (cf.~\cite{rabinovich}) stating that $\LTL$ and $\FO$ have the same expressive power over $\Sigma^*$. But the translation of an $\FO$-sentence to an $\LTL$-sentence can involve a non-elementary blow-up in size. Indeed, \cite{drive} proves lower bounds on the length of $\LTL$-sentences expressing parts of Regulation 561.  Very large sentences are not natural and lead to prohibitive model-checking times.

\begin{example} To illustrate the point, consider the following law in Regulation 561:
\begin{quote}\sf
Article 6.2: The weekly driving time shall not exceed 56 hours\ldots
\end{quote}
Restrict attention to words representing one week of activities, i.e., words of length $7\cdot 24\cdot 60$ over the alphabet $\Sigma=\{d,w,r\}$. A straightforward formalization of Article 6.2 in $\LTL$ is (using $d\in\Sigma$ as a propositional variable) the huge disjunction of 
$$ \bigwedge_{j\le D}\Big(\bigwedge_{r_{j}\le i<\ell_{j+1}}\Circle^i\neg d\wedge \bigwedge_{\ell_j\le i< r_{j}}\Circle^id\Big)$$
for all $D\le 7\cdot 24\cdot 60$ and all $r_0:=0\le \ell_1< r_1<\cdots <\ell_D<r_D<\ell_{D+1}:=7\cdot 24\cdot 60$ with $\sum_{1\le j\le D}(r_j-\ell_j)\le 56\cdot 60$. These are 
$> {7\cdot 24\cdot 60 \choose 56\cdot 60} > 10^{2784}$ many disjuncts.
\end{example}

To conclude, $\MSO$ gives the wrong model because it does not allow sufficiently fast model-checkers, and $\LTL$ is the wrong model because it is not sufficiently (expressive nor) succinct, hence not natural. It can be expected that, like Regulation 561, many algorithmic laws concerning sequences of activities state lower and upper bounds on the duration of certain activities or types of activities. The constants used to state these bounds are not necessarily small, and an important aspect to take into account when analyzing the model-checking complexity.

\subsection{Regulation 561 and timed modal logics}\label{sec:timedmodal}

The above motivates to look at models with built-in timing constraints: ``In practice one would want to use ‘sugared’ versions of $\LTL$, such as metric temporal logic ($\MTL$; \cite{ow}) which allow for expressions such as $\Circle^{n+1}$ to be represented succinctly''\cite{drive}. $\MTL$ has modalities like $\Diamond_{[5,8]}\varphi$ expressing that $\varphi$ holds within 5 and 8 time units from now.

For Regulation 561,
cases are tachograph recordings which, formally, are {\em timed words} $ (a_0, t_0)\ (a_1,t_1)\cdots $ where the $a_i$ are letters and the $t_i$ an increasing sequence of time-points; intuitively, activity $a_0$ is observed until time point $t_0$, then $a_1$ until $t_1$, and so on. Alur and Dill \cite{alurdill} extended finite automata to {\em timed automata} that accept sets of timed words  -- see \cite{bouyerTA} for a survey. 
Roughly speaking, computations of such automata happen  in time and are governed by finitely many clocks: transitions from one state to another are enabled or blocked depending on the clock values, and transitions can reset  some clocks (to value 0). Alur and Dill
\cite{alurdill} proved that timed automata have decidable emptiness, thus enabling the dominant model-checking paradigm. 

Consequently, a wealth of timed temporal logics have been investigated -- \cite{Hsurvey,bouyer} are  surveys. 
 The following are some of the most important choices when defining such a logic:
 \begin{center}
\begin{tabular}{ll | ll | l}
{\em semantics} &   &  {\em time}                         &  & {\em clocks}\\\hline
finite words & signal-based & continuous $\mathbb{R}_{\ge 0}$ &  branching   &  internal    \\
infinite words & event-based & discrete $\mathbb{N}$  &  linear         &   external 
\end{tabular}
\end{center}
A subtle choice is between signal- or event-based semantics. It means, roughly and respectively,  that the modalities quantify  over all time-points or only over the $t_i$ appearing in the timed word; $\MTL$ is known to be less expressive in the latter semantics over finite timed words \cite{mtlpwcont}.
 A crucial choice is between time $\N$ or $\R_{\ge0}$. Internal clocks appear only on the side of the automata, external clocks appear in sentences which reason about their values. We briefly survey the most important results. 

An early success  \cite{acd} concerns the infinite word signal-based branching continuous time logic $\TCTL$ (timed computation tree logic): over (systems modeled by) timed automata it admits a model-checker  with runtime $t^{O(c)}\cdot k\cdot n$ where $n$ is the automaton size, $k$ the size of the input sentence,~$c$ the number of clocks, and $t$ is the largest time constant appearing in the input. \cite{hnsy} extends this allowing external clocks. However, continuous branching time is semantical and syntactical overkill for Regulation 561. For linear continuous time we find $\MTL$ and $\TPTL$ (timed propositional temporal logic), a more expressive~\cite{bouyerexpressive} extension with external clocks. Since model-checking is undecidable for these logics  \cite{ah94,afh}, fragments have been investigated. Surprisingly \cite{ow} found an fpt model-checker for $\MTL$ over event-based finite words via a translation to  alternating automata with one clock, albeit with intolerable parameter dependence (non-primitive recursive).  $\MITL$ (metric interval temporal logic) \cite{afh} is the fragment of $\MTL$ disallowing singular time constraints as, e.g., in $\Diamond_{[1,1]}\varphi$. \cite{mitltota,mitl2ta} gives an elegant translation of $\MITL$ to  timed automata and thereby a model-checker with runtime\footnote{In fact, $t$ can be replaced by a typically smaller number, called the {\em resolution} of the formula -- see \cite{mitltota}.} $2^{O(t\cdot k)}\cdot n$. Over discrete time,~\cite{ah94} adapts  the mentioned translation of $\LTL$ to B\"uchi automata and gives a model-checker for $\TPTL$ with runtime $2^{O(t^c\cdot k)}\cdot n$.

As said, from the perspective of algorithmic law, $t$ is not typically small and runtimes exponential in $t=56h=3360\textit{min}$ are prohibitive. Tamer runtimes with~$t$ moved out of the exponent 
have been found for a certain natural $\MITL$-fragment $\MITL_{0,\infty}$ both over discrete and continuous time -- see \cite{Hsurvey,afh}. 

However, ``standard real-time temporal logics [\ldots] do not allow us to constrain the accumulated satisfaction time of state predicates'' \cite[p.414]{achduration}. It seems that this is just what is  required to formalize the mentioned Article 6 (2), and we expect similar difficulties to be encountered with other laws concerning activity sequences.

There are various attempts to empower the logics with some reasoning about durations. {\em Stopwatch automata} \cite{dima} are timed automata that can not only reset clocks but also stop and activate them. However, emptiness is undecidable already for a single stopwatch~\cite{hkpv}. Positive results are obtained in  \cite{achduration} for {\em observer stopwatches}, i.e.,  roughly, stopwatches not used to govern the automaton's transitions.  On the logical side,  \cite{boua} and \cite{observer} study fragments and restrictions for $\TCTL$ with (observer) stopwatches.   On another strand, \cite{durcalc} puts forward the {\em calculus of durations}, but already tiny fragments turn out undecidable \cite{chao}. For discrete time, \cite{hansen} gives an fpt model-checker via a translation to finite automata. For continuous time, \cite{fraenzle} obtains fpt results under certain reasonable restrictions of the semantics. A drawback is that these fpt results have non-elementary parameter dependence.

To conclude, the extensive research on ```sugared' versions'' of $\LTL$ in system verification does not reveil a good answer to our central question for a model-checking problem modeling algorithmic laws concerning activity sequences. In particular, many known model-checkers are too slow in that they do not scale well with time constants mentioned in the law.

\subsection{The perspective from algorithmic law}
 
The new perspective on model-checking from algorithmic law seems orthogonal to the  dominant perspectives from database theory and system verification in the sense that it seems to guide incomparable research directions.
 
 In database theory there is special interest in model-checking problems for a rich class~$\cal K$, formalizing a large class of databases, and possibly weak logics $L$ formalizing simple basic queries. In algorithmic law (concerning activity sequences) it is the other way around, focussing on $\cal K=\Sigma^*$. 

System verification gives special interest to infinite words and continuous time (cf.\ e.g.\ \cite{acd})  while algorithmic law focusses on finite words and discrete time.
Most importantly, system verification focusses on structures specifying {\em sets} of words: its model-checking problem corresponds to (a generalization of) the consistency problem in algorithmic law. 
In algorithmic law the consistency problem is secondary, the main interest is in evaluating sentences over single words. 

Finally, the canonical parameterization of a model-checking problem takes the size~$\|\varphi\|$ of the input sentence $\varphi$  as the parameter. Intuitively, then parameterized complexity analysis focusses attention  on inputs of the problem where $\|\varphi\|$ is relatively small. Due to large constants on time constraints appearing in the law to be formalized this parameterization does not seem to result in a faithful model of algorithmic law. We shall come back to this point in Section \ref{sec:pswa}. 

 Compared to system verification this shift of attention in algorithmic law
opens the possibility to use more expressive logics while retaining tractability of the resulting model. 
In particular, complexity can significantly drop via the shift from continuous time, infinite words and consistency-checking, to discrete time, finite words and model-checking. 
While discrete time is well investigated in system verification, it has been noted that both finite words and model-checking have been neglected -- see \cite{fionda} and~\cite{schnoebelenpath}, respectively.
 To make the point: over finite words consistency-checking $\LTL$ is $\PSPACE$-complete but
 model-checking is $\PTIME$,  even for the more succinct extensions of $\LTL$ with past- and now-modalities \cite{schnoebelenpath}, or even finite variable $\FO$~\cite{vardifok}.\footnote{3-variable (2-variable) $\FO$ has the same expressive power as (unary) $\LTL$ over finite words but is much more succinct \cite{evw}. \cite{savitch} gives a fine calibration of the parameterized complexity of finite variable $\FO$.}
 
 \subsection{Model-checking stopwatch automata: summary}\label{sec:techsum}

We take advantage of this possibility to use more expressive logics and suggest (a version of) discrete time {\em stopwatch automata} $\SWA$ as an answer to our central question, that is, we propose $\textsc{MC}(\Sigma^*,\SWA)$ as a model for algorithmic laws concerning sequences of activities.\footnote{
In the notation of Section~\ref{sec:compprobl} we define $w\models\A$ for a finite word $w$ and a stopwatch automaton $\A$ to mean that $\A$ accepts $w$.} 
Our stopwatches are bounded, and their bounds correspond to time constants mentioned in laws. 
Mimicking notation from Section~\ref{sec:timedmodal}, we  let $c_\A$ denote the number of stopwatches and~$t_\A$ the largest stopwatch bound of a stopwatch automaton~$\A$.
We give the following upper bound on the complexity of $\textsc{MC}(\Sigma^*,\SWA)$:

\begin{theorem}\label{thm:intromodcheck} There is an algorithm that given 
a stopwatch automaton $\A$ and a word~$w$ decides whether $\A$ accepts $w$ in time 
$$
O\big(\|\A\|^2\cdot t_\A^{c_\A}\cdot|w|\big).
$$ 
\end{theorem}

We prove a slightly stronger result in Theorem~\ref{thm:modcheck}. Notably, the aspect $t_\A$ does not appear in the exponent, so this overcomes a bottle-neck of various model-checkers designed in system verification (see Section~\ref{sec:timedmodal}). We obtain similar algorithms for consistency-checking  and scheduling (Corollary~\ref{cor:conscheck} and Theorem~\ref{thm:scheduling}). This is despite the fact that stopwatch automata are highly expressive, namely have the same expressive power as {\sf MSO} over finite words (Theorems~\ref{thm:regular} and \ref{thm:buechi}).

The final Section~\ref{sec:concl} discusses our model $\textsc{MC}(\Sigma^*,\SWA)$ following the criteria of Section~\ref{sec:formalmodel}, and gives a critical examination of the factor $t_\A^{c_\A}$ in the runtime of our model-checker. Intuitively, typical inputs have small $c_\A$ and large $t_\A$, and it would be desirable to replace this factor by, e.g., $2^{O(c_\A)}\cdot t_\A^{O(1)}$. We show this is unlikely to be possible. Theorem~\ref{thm:pswa} implies:

\begin{theorem}\label{thm:introlower} Assume $\FPT$ does not contain the W-hierarchy. Let $f:\N\to\N$ be a computable function. Then there does not exist an algorithm that given 
a stopwatch automaton $\A$ and a word $w$ decides whether $\A$ accepts $w$ in time 
$$
\big(\|\A\|\cdot f(c_\A)\cdot  t_\A\cdot |w|\big)^{O(1)}.
$$
\end{theorem}
The complexity-theoretic assumption here is weaker than $\mathsf{FPT}\neq\mathsf{W}[1]$ considered earlier.

\section{Stopwatch automata}\label{sec:stopwatch}

Before giving our definition we informally describe the working of a {\em stopwatch automaton}. A stopwatch automaton is an extension of a finite automaton whose computations happen in discrete time: the automaton can stay for some amount of time in some state and then take an instantaneous transition to another state. 

There are constraints on which transitions can be taken at a given point of time as follows. Time is recorded by a set of {\em stopwatches} $X$, every stopwatch $x\in X$ has a {\em bound}~$\beta(x)$, a maximal time it can record. Every  stopwatch is {\em active} or not in a given state. 
During a run that stays in a given state for a certain amount of time, the value of the active stopwatches increases by this amount of time (up to their bounds) while the 
inactive
stopwatches do not change their value. Transitions between states are labeled with a {\em guard} and an {\em action}. The guard is a condition on the values of the stopwatches that has to be satisfied for the transition to be taken, usually requiring upper or lower bounds on certain stopwatch values. The action modifies stopwatch values, for example, resets some of the stopwatches to value~0.

Instead of transitions, states are labeled by letters of the alphabet.
A stopwatch automaton {\em accepts} a given word if there exists a computation leading from a special state {\em start} to a special state {\em accept} and that 
 {\em reads} the word: staying in a state for 5 time units means reading 5 copies of the letter labelling the state.

\subsection{Abstract stopwatch automata}

We now give the definitions that have been anticipated by the informal description above.

\begin{definition}
An {\em abstract stopwatch automaton} is a tuple $\A=(Q,\Sigma, X, \lambda,\beta, \zeta,\Delta)$ where
\begin{enumerate}\itemsep=0pt
\item[--] $Q$ is a finite set of {\em states} containing the states {\em start} and {\em accept};
\item[--] $\Sigma$ is a finite {\em alphabet};
\item[--] $X$ is a finite set of {\em stopwatches};
\item[--] $\lambda:Q\to\Sigma$;
\item[--] $\beta:X\to \N$ maps every stopwatch $x\in X$ to its {\em bound} $\beta(x)\in\N$;
\item[--] $\zeta\subseteq X\times Q$ contains pairs $(x,q)$ such that
the stopwatch $x$ is  {\em active in} state $q$;
\item[--] $\Delta\subseteq 
Q\times\G\times \Act\times Q$ is a set of {\em transitions}.
\end{enumerate}
Here, $\G$ is the set of {\em abstract guards}  (for $\A$), namely 
sets of assignments, and $\Act$ is the set of {\em abstract actions}  (for $\A$), namely functions from assignments to assignments. An {\em assignment}   (for $\A$) is a function $\xi:X\to \N$ such that $\xi(x)\le\beta(x)$ for all $x\in X$. To be precise, we should speak of a $\beta$-assignment but the $\beta$ will always be clear from the context.
We define the {\em bound of $\A$} to be 
$$
B_\A:=\prod_{x\in X}(\beta(x) +1)
$$ 
understanding that the empty product is 1 so that $\prod_{x\in \emptyset}(\beta(x) +1) :=1$. 
This  is the cardinality of the set of assignments (for $\A$).
We say that a transition $(q,g,\alpha,q')\in\Delta$ is {\em from} $q$, and {\em to} $q'$, and {\em has} abstract guard $g$ and abstract action $\alpha$.

\end{definition}

Computations of stopwatch automata are defined in terms of their corresponding {\em transition systems}.


\medskip

\begin{definition}\label{defi:TSA}
Let  $\A = (Q,\Sigma, X, \lambda,\beta, \zeta,\Delta)$ be an abstract stopwatch automaton. The {\em transition system $\TS(\A)$} of  $\A$ is  given by a set of nodes and labeled edges: a {\em node (of $\TS(\A)$)} is a pair $(q,\xi)$ of a state $q\in Q$ and an assignment $\xi$;  a {\em labeled edge (of $\TS(\A)$)} is a triple $\large((q,\xi),t,(q',\xi')\large)$ for nodes $(q,\xi),(q',\xi')$ and $t\in \N$  such that {\bf either}:
\begin{enumerate}\itemsep=0pt
\item[--] $t=0$ and $(q,g,\alpha,q')\in \Delta$  for an abstract guard $g$ and an abstract action $\alpha$ such that $\xi\in g$ and $\alpha(\xi)=\xi'$, \\\\
{\bf or},\\
\item[--] $t>0$ and $q=q'$ and $\xi'$ is the assignment given by
$$
\xi'(x)=\left\{
\begin{array}{ll}
\min\big\{\xi(x)+t,\beta(x)\big\} &\text{if } (x,q)\in \zeta,\\ 
\xi(x)&\text{else.}
\end{array}
\right.
$$
\end{enumerate}
\end{definition}

For $t\in\N$ we let $\stackrel{t}{\to}$ be the binary relation that contains those pairs $\large((q,\xi),(q',\xi')\large)$ of nodes such that $\large((q,\xi),t,(q',\xi')\large)$ is a labeled edge.

\begin{definition}
A (finite) {\em computation of $\A$} is a finite walk in $\TS(\A)$, i.e., for some $\ell\in\N$ a sequence 
$$
\Big(\big((q_i,\xi_i),t_i,(q_{i+1},\xi_{i+1})\big)\Big)_{i<\ell}
$$ of directed edges of $\TS(\A)$ such that $q_i\neq\mathit{accept}$ for all $i<\ell$; we write this as
$$
(q_0,\xi_0)\stackrel{t_0}{\to}(q_1,\xi_1)\stackrel{t_1}{\to}(q_2,\xi_2)\stackrel{t_2}{\to}\cdots
\stackrel{t_{\ell-1}}{\to}(q_\ell,\xi_\ell).
$$
In this case, we say that the computation is  {\em from} $(q_0,\xi_0)$ and {\em to} $(q_\ell,\xi_\ell)$; 
it is {\em initial} if $\xi_0$ is constantly~0 and $q_0=\mathit{start}$;
it is {\em accepting} if $q_\ell=\textit{accept}$. 
The computation {\em reads} the word $$\lambda(q_0)^{t_0} \lambda(q_1)^{t_1}\cdots \lambda(q_{\ell-1})^{t_{\ell-1}}.$$
\end{definition}

We understand that $\sigma^0$ denotes the empty string for every letter $\sigma$ in the alphabet $\Sigma$ and juxtaposition of strings corresponds to concatenation. Through computations, we define strings and languages accepted by a Stopwatch automaton.

\begin{definition}
The automaton
$\A$ {\em accepts} $w\in\Sigma^*$ if there is an initial accepting computation of $\A$ that reads~$w$. The set of these words is the {\em language $L(\A)$ of $\A$}.
\end{definition}

\begin{remark} The requirement  $q_i\neq\mathit{accept}$ for all $i<\ell$ in the definition of computations means that we interpret $\mathit{accept} $ as a halting state;  it implies that the label $\lambda(\mathit{accept})$ as well as transitions from {\em accept} are irrelevant. Without this condition, $w\in L(\A)$ implies $wa^n\in L(\A)$ for $a:=\lambda(\mathit{accept})$ and all $w\in\Sigma^*$ and $n\in\N$.\end{remark}

\begin{remark} 
Stopwatch automata are straightforwardly explained for continuous time $\R_{\ge 0}$ where they read timed words, and bounds $\beta(x)=\infty$. Stopwatch automata according to
\cite{dima,hkpv}  are such automata where
guards are  Boolean combinations of $x\ge c$ (for $x\in X$ and $c\in\N$), and actions are resets (to 0 of some stopwatches). The emptiness problem for those automata is undecidable \cite{hkpv}. 
So-called \emph{timed automata}  additionally require stopwatches to be active in all states, and have decidable emptiness \cite{alurdill}. The model allowing $x\ge  c+y$ (for $x,y\in X$ and $c\in\N$) in guards still has decidable emptiness and is exponentially more succinct than guards with just Boolean combinations of $x\ge c$ (\cite{diagonalsuccinct}). Allowing more actions is subtle, e.g., emptiness becomes undecidable when 
$x:=x\dotminus1$ or when $x:=2x$
is allowed;  see~\cite{update} for a detailed study. 
\end{remark}

\subsection{Specific stopwatch automata}

To consider an abstract stopwatch automata as an input to an algorithm, we must agree on how to specify the guards and actions, i.e., properties and functions on assignments. This is a somewhat annoying issue because on the one hand our upper bounds on the model-checking complexity turn out to be robust with respect to the choice of this specification in the sense that they scale well with the complexity of computing guards and  actions, so a very general definition is affordable. On the other hand, for natural stopwatch automata including the one we are going to present for the European Traffic Regulation 561, we expect guards and actions to be simple properties and functions. 

As mentioned, typically guards mainly compare certain stopwatch values with constants or other values, and actions do simple re-assignments of values like setting some values to~0. 
Hence our choice on how to specify guards and actions is somewhat arbitrary. To stress the robustness part, we use a general model of computation: Boolean circuits. In natural automata, we expect these circuits to be small. 


An assignment determines for each stopwatch $x\in X$ its bounded value and as such can be specified by 
\[
b_\A:=\sum_{x\in X}\lceil\log(\beta(x)+1)\rceil
\] 
many bits. 
We think of the collection of~$b_\A$ bits as being composed of blocks, with a block of $\lceil\log(\beta(x)+1)\rceil$ bits corresponding to the binary representation of the value of stopwatch $x\in X$ under the assignment. 
A {\em specific guard} is a Boolean circuit with one output gate and~$b_\A$ many input gates. 
A specific guard determines an abstract one in the obvious way.

A {\em specific action} is a Boolean circuit with $b_\A$ many output gates  and  $b_\A$ many input gates. 
On input an assignment, for each clock $x\in X$, it computes the binary representation of a value $v_x\in \N$ in the block  of $\lceil\log(\beta(x)+1)\rceil$ output gates corresponding to $x$. 
Furthermore, we agree that the assignment computed by the circuit maps $x$ to $\min\{v_x,\beta(x)\}$ thereby mapping assignments to assignments. 
A specific action determines an abstract one in the obvious way.

A {\em specific stopwatch automaton} is defined like an abstract one but with specific guards and actions replacing abstract ones. A specific stopwatch automaton determines an abstract one taking the abstract guards and actions as those determined by the specific ones. Computations  of specific stopwatch automata and the language they accept are defined as those of the corresponding abstract one. The {\em size} $\|\A\|$of specific stopwatch $\A$ automaton is the length of a reasonable binary encoding of it.  

We shall only be concerned with specific stopwatch automata and shall mostly omit the qualification `specific'.

\subsection{A definitorial variation}

To showcase the robustness of our definition and for later use, 
 we mention a natural variation of our definition and show it is inessential. 
 
 Define a {\em $P(\Sigma)$-labeled (specific) stopwatch automaton} $\A=(Q,\Sigma,X,\lambda,\beta,\zeta,\Delta)$ like a (specific) stopwatch automaton but with $\lambda:Q\to P(\Sigma)\setminus\{\emptyset\}$. A computation
\begin{eqnarray}\label{eq:compA}
&&(q_0,\xi_0)\stackrel{t_0}{\to}(q_1,\xi_1)\stackrel{t_1}{\to}(q_2,\xi_2)\stackrel{t_2}{\to}\cdots
\stackrel{t_{\ell-1}}{\to}(q_\ell,\xi_\ell).
\end{eqnarray}
is said to {\em read} any word $a_0^{t_0}\cdots a_{\ell-1}^{t_{\ell-1}}$ with $a_i\in\lambda(q_i)$ for every $i<\ell$. The language $L(\A)$ of $\A$ is defined as before.

A stopwatch automaton  can be seen as a
 $P(\Sigma)$-labeled stopwatch automaton whose state labels are singletons.
  Conversely, given a $P(\Sigma)$-labeled stopwatch automaton $\A=(Q,\Sigma,X,\lambda,\beta,\zeta,\Delta)$ we define a stopwatch automaton $\A'=(Q',\Sigma,X,\lambda',\beta,\zeta',\Delta')$ as follows: 
its states $Q'$ are pairs $(q,a)\in Q\times \Sigma$ such that $a\in\lambda(q)$; for the start and accept states choose any $(q,a)$ for $q$ the start, resp., accept state of $\A$. The  $\lambda'$-label of $(q,a)\in Q'$ is $a$ and
 stopwatch $x\in X$ is active (according $\zeta'$) in $(q,a)$ if and only if it is active in $q$ (according~$\zeta$). 
  We let $\Delta'$ contain a transition $((q,a),g,\alpha,(q',  a'))$ if $(q,a),(q', a')\in Q'$ and 
$(q,g,\alpha, q')\in\Delta$. Further, we add transitions with trivial guards and actions from $(q,a)\in Q'$ to $(q,a')\in Q'$.

Given a computation of $\A$ as above, choose any $a_i\in\lambda(q_i)$ for every $i<\ell$. Then 
\begin{eqnarray}\label{eq:compA'}
&&
((q_0,a_0),\xi_0)\stackrel{t_0}{\to}((q_1,a_1),\xi_1)\stackrel{t_1}{\to}((q_2,a_2),\xi_2)\stackrel{t_2}{\to}\cdots
\stackrel{t_{\ell-1}}{\to}((q_\ell,a_\ell),\xi_\ell).
\end{eqnarray}
is a computation of $\A'$. The choice of the $a_i$ can be made so that this computation reads the same word as the computation \eqref{eq:compA}. If \eqref{eq:compA} is initial (accepting), make \eqref{eq:compA'} initial (accepting) by adding a $\stackrel{0}{\to}$-transition from (to) the start (accept) state of $\A'$. 

Conversely, if \eqref{eq:compA'} is a computation of $\A'$, then \eqref{eq:compA} is a computation of $\A$ that reads the same word. 
To sum up: 

\begin{proposition} \label{prop:plambda} There is a polynomial time computable function that maps  every $P(\Sigma)$-labeled 
 stopwatch automaton  $\A$ to a stopwatch automaton 
  $\A'$ with $B_{\A'}=B_{\A}$ and $L(\A)=L(\A')$.
\end{proposition}

\section{A stopwatch automaton for Regulation 561}\label{sec:automaton}

Aside expressivity and tractability, we stressed naturality as a criterion of models for algorithmic law.
In this section and the next section we make the point for stopwatch automata by implementing Regulation 561. 
As already mentioned, Regulation 561 is
 a complex set of articles concerning sequences of activities of truck drivers. Possible activities are {\em driving}, {\em resting} or {\em other work}.  
 The activities over time are recorded by tachographs and formally understood as words over the alphabet $\Sigma:=\{d,r,w\}$. In the real world time units are minutes. Regulation 561 
limits driving and work times by demanding breaks, daily rest periods and weekly rest periods, both of which can be regular or reduced under various conditions.


We construct a stopwatch automaton that accepts precisely the words over~$\Sigma$ that represent activity sequences that are legal according to Regulation 561. The states~$Q$ of the automaton are:
\begin{quote}\it
drive,\ \ break, \ \ other work,\ \  

reduced daily, \ \ regular daily, \ \ reduced weekly,\ \  regular weekly, 

compensate1, \ \ compensate2,  \ \ week, \ \ start, \ \  accept. 
\end{quote}
The states in the first row have the obvious meaning. The states in the second row represent different kinds of rest periods. The function $\lambda$ labels {\em other work} by $w$, {\em drive} by~$d$ and all other states by~$r$. The  states {\em compensate1} and {\em compensate2}  are used for the most complicated part of Regulation 561 that demands certain compensating rest periods whenever a weekly rest period is reduced. The state  {\em week} is auxiliary, and 
accepting computations spend 0 time in it. The same is true for {\em start}. So, the $\lambda$-labels of {\em start} and {\em week} do not matter.



We construct the automaton stepwise implementing one article after the next, introducing stopwatches along the way. For each stopwatch $x$ we state its bound $\beta(x)$ and the  states $q$ in which  it is active, i.e., specifying the pairs $(x,q)\in \zeta$. We shall refer to stopwatches that are nowhere active as {\em counters} or {\em registers}, depending on their informal usage; a {\em bit} is a counter with bound 1.

We describe a transition $(q,g,\alpha,q')$ saying that there is a transition from $q$ to $q'$ with guard $g$ and action $\alpha$. We specify guards by a list of expressions of the form $z\le r$ or $z+z'>r$ or the like for $r\in \N$; this is shorthand for a circuit that checks the conjunction of these conditions. We specify actions by  lists of expressions of the form $z:=r$ or $z:= z'+r$ or the like for $z, z' \in X$ and $r\in\N$; this is shorthand for the action that carries out the stated re-assignments of values in the order given by the list. 
These lists are also described stepwise treating one article after the next. As a mode of speech, when treating a particular law, we shall say that a given transition has this or that action or guard: what we mean is that the  actions or guards of the transition of the final automaton is given by the lists of these statements in order of appearance (mostly the order won't matter).

\medskip

We illustrate this mode of speech by describing the automaton around {\em start}:
let~$x_{\mathit{start}}$ be a stopwatch with bound 1 and active at {\em start}; there are no transitions to {\em start} and transitions from {\em start} to all other states except {\em week}; these transitions have guard
 $x_{\mathit{start}}=0$. 
 Later these transitions shall get more guards and also some actions. These stipulations  mean more precisely the following:
the bound $\beta$ satisfies  $\beta(x_{\mathit{start}})=1$;
 the  set  $\Delta$ contains for any state $q\notin\{\textit{week},\textit{start}\}$ the transition $(\textit{start}, g,\alpha,q)$ where the guard $g$ checks the conjunction of $x_{\mathit{start}}=0$ and the other guards introduced later, and the action $\alpha$ carries out the assignments and re-assignments as specified later; further, $(x_{\mathit{start}},q)\in\zeta$
  if and only if $q=\textit{start}$. 


\medskip

We loosely divide Regulation 561 into daily and weekly demands. We first describe how to implement the daily demands using the first 5 states and {\em daily driving} and {\em accept}. The other states will be used to implement the weekly demands. 

During the construction we shall explicitly collect the constants appearing in the articles and denote them by $t_0,\ldots, t_{16}$. Our construction is such that these constants determine all guards, actions and bounds in an obvious way. Knowing this will be useful for the discussion in later sections.

\subsection{Daily demands}

We use the first 3 states to implement the the law about {\em continuous driving}:
\begin{quote}\sf
Article 7 (1st part): After a driving period of four and a half hours a driver shall take an uninterrupted break of not less than 45 minutes, unless he takes a rest period.
\end{quote}

We use a stopwatch $x_{cd}$ with bound $4.5h+1=271$ that is active in {\em drive}. Further, we use a stopwatch $x_{\textit{break}}$ with bound $9h$ that is active in {\em break}. 
For the law under consideration we could use the bound of $4.5h+1$, the reason we use $9h$ will become clear later when implementing Article 8.7.

There are transitions back and forth between any two of the states {\em break}, {\em drive} and {\em other work}. We give the transitions to {\em break} action $x_{\textit{break}}:=0$, and the transitions from {\em drive} the guard  $x_{cd}\le 4.5h$. This ensures that a computation  staying in {\em drive} for more than 4.5h will not be able to leave this state, so cannot be accepting. We add two transitions from {\em break}  to both {\em drive} and {\em other work} with guard $x_{\textit{break}}\ge 45$ and action $x_{\textit{break}}:=0;\  x_{\textit{cd}}:=0$.


Transitions to {\em regular daily} and {\em reduced daily} have action $x_{cd}:=0$: this ensures
the ``unless\ldots'' statement in Article 7 (transitions to weekly rest periods described below will also have this action). The first part of this Article 7 uses constants $t_0:=4.5h=270;\  t_1:=45$ (the constant $9h$ is denominated later by $t_{16}$). 

\begin{figure}[ht]
\centering
\begin{tikzpicture}[accepting/.style={double distance=3pt}]
\node[state] (s0) {$\begin{array}{cc}\mathit{drive}\\ x_{\textit{cd}}\end{array}$};
\node[state, right of=s0,xshift=11cm] (s1) {$\begin{array}{cc}\mathit{other}\\ \mathit{work}
\end{array}$};
\node[state, right of=s0,xshift=5cm,yshift=4cm] (s2) {$\begin{array}{cc}\mathit{break}\\ x_{\textit{break}}
\end{array}$};

\draw   
(s0) edge[bend left, left] node{$\begin{array}{c}x_{\textit{cd}}\le 270\\x_{\textit{break}}:=0\end{array}$} (s2)
(s2) edge[above] node{} (s0)
(s2) edge[bend left,right] node{\hspace{-2ex}$\begin{array}{c}x_{\textit{break}}\ge 45\\x_{cd}:=0\end{array}$} (s0)
(s2) edge[above] node{} (s1)
(s2) edge[bend right,left ] node{$\begin{array}{c}x_{\textit{break}}\ge 45\\x_{cd}:=0\end{array}$} (s1)
(s1) edge[bend right, right] node{$x_{\textit{break}}:=0$} (s2)
(s2) edge[above] node{} (s0)
(s0) edge[bend right, above] node{$x_{\textit{cd}}\le 270$} (s1)
(s1) edge[ above] node{} (s0)
;
\end{tikzpicture}
 \caption{Illustration of  Article 7 (first part);  stopwatches $x_{\textit{cd}},x_{\textit{break}}$ are shown at the states where they are active.}
\end{figure}
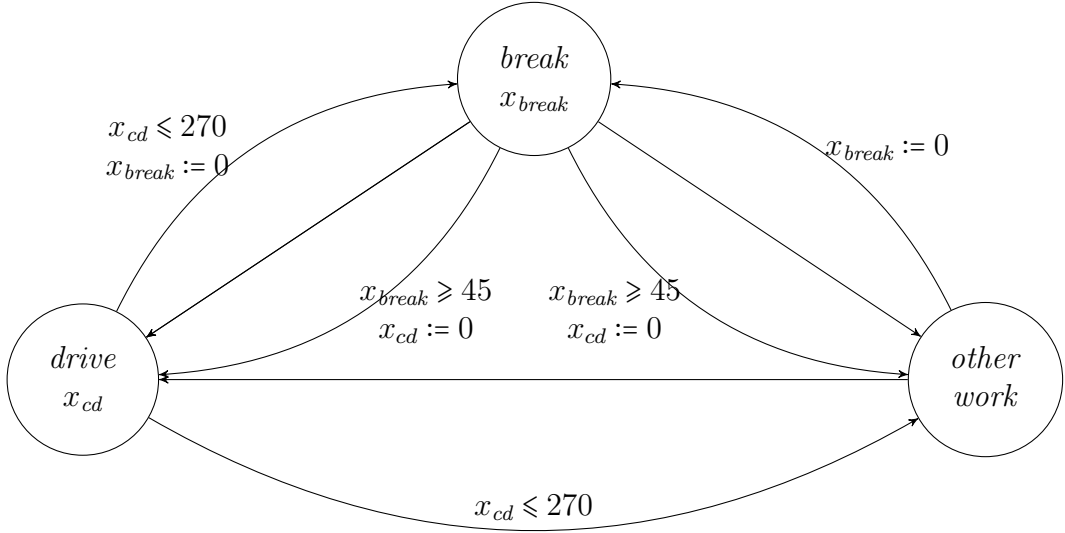

\noindent Article 7 allows to divide the demanded break into two shorter ones:
\begin{quote}\sf
Article 7 (2nd part): This break may be replaced by a break of at least 15 minutes followed by a break of at least 30 minutes each distributed over the period in such a way as to comply with the provisions of the first paragraph.
\end{quote}

To implement this possibility, we use a bit $b_{rb}$ that, intuitively, indicates a reduced break. We add  transitions from {\em break} to {\em other work} and {\em drive} with guard $15\le x_{\textit{break}}<45$ and action $b_{rb}:=1;\ x_{\textit{break}}:=0$. We note that these transitions do not have action $x_{cd}:=0$. 
We add transitions
 from {\em break} to {\em other work} and {\em drive} with guards $b_{rb}=1$ and $30\le x_{\textit{break}}$ and action $b_{rb}:=0;\ x_{cd}:=0; \ x_{\textit{break}}:=0$. 
Transitions to states representing daily or weekly rests introduced below all get action $b_{rb}:=0$.
The second part of Article 7 uses the constant $t_2:=15$; we do not introduce a name for 30 but view this constant as equal to $t_1-t_2=45-15$.\medskip


\noindent Article 4.(k) defines `daily driving time' as the accumulated driving time between two daily rest periods. 
According to Article 4.(g)  daily rest periods can be regular or reduced, the former meaning at least $11h$ of rest, the latter means less than $11h$ but at least $9h$ of rest. These are represented by the states {\em regular daily} and {\em reduced daily}.
\begin{quote}\sf
Article 8.1: A driver shall take daily and weekly rest periods.

Article 8.2: Within each period of 24 hours after the end of the previous daily rest period or weekly rest period a driver shall have taken a new daily rest period. If the portion of the daily rest period which falls within that 24 hour period is at least nine hours but less than 11 hours, then the daily rest period in question shall be regarded as a reduced daily rest period.

\end{quote}

Weekly rest periods are treated in the next subsection. We use a stopwatch $x_{\textit{day}}$ with bound $24h+1$ which is active in all states except {\em accept} and {\em start},
 and a stopwatch $x_{dr}$ with bound $11h$ active in {\em reduced daily} and {\em regular daily}. We have transitions back and forth between the states {\em break}, {\em drive}, {\em other work} and the states {\em regular daily}, {\em reduced daily}. The transitions to {\em regular daily} are guarded by $ x_{\textit{day}}\le 24h-11h=780;\ b_{rb}=0$; transitions to  {\em reduced daily} are guarded by $x_{\textit{day}}\le 24h-9h=900;\ b_{rb}=0$. 
The transitions from 
{\em regular daily} are guarded by $x_{dr}\ge 11h$, and the transitions from {\em reduced daily} are guarded by $11h>x_{dr}\ge 9h$ -- later we shall refer to these guards as  {\em definitorial} for the states {\em regular daily} and {\em reduced daily}. Transitions from {\em regular daily}, {\em reduced daily} have action $x_{dr}:=0, x_{\textit{day}}:=0$. 

All transitions to {\em accept} get guard $x_{\textit{day}}\le 24h$. Note that an accepting computation cannot involve an assignment satisfying $x_{\textit{day}}>24h$, so eventually has to visit and leave {\em regular daily} or {\em reduced daily} (or their weekly counterparts, see below). This ensures  Article 8.1  for daily rest periods. These laws use constants
$t_3:=24h=1440, t_4:=11h=660, t_5:=9h=540$.

\medskip

\noindent Actually, the definition of regular daily rest periods in  Article 4.(g)  is more complicated:
\begin{quote}\sf
`regular daily rest period' means any period of rest of at least 11 hours. Alter\-natively, this regular daily rest period may be taken in two periods, the first of which must be an uninterrupted period of at least 3 hours and the second an un\-in\-ter\-rupted period of at least nine hours,
\end{quote}

To implement this we use a bit $b_{dr}$ indicating that a $3h$ part of a regular daily rest period has been taken. We duplicate the transitions from {\em regular daily} but replace the guard $x_{dr}\ge 11h$ by $x_{dr}\ge 9h,b_{dr}=1$.  To add the possibility of taking a partial regular daily rest period of at least $3h$ we add transitions
from {\em regular daily} to {\em drive} and {\em other work} with 
guards $b_{dr}=0,3h\le x_{dr}<11h$ and action $b_{dr}:=1$; note these transitions do not have action
$x_{\textit{day}}:=0$. All transitions with action $x_{\textit{day}}:=0$ also get action $b_{dr}:=0$, including those modeling weekly rest periods described below. This uses the constants $t_6=3h=180,t_7:=9h=540$.

\medskip

\noindent The final daily demand constrains daily driving times:

\begin{quote}\sf
Article 6.1: The daily driving time shall not exceed nine hours. However, the daily driving time may be extended to at most 10 hours not more than twice during the week.
\end{quote}

To implement Article 6.1 we use a stopwatch $x_{dd}$ active at {\em drive} with bound $10h+1$ to measure the daily driving time. 
Additionally,  we use a counter $c_{dd}$ with bound $3$. As described later, this counter will be reset to 0 when the week changes. Duplicate the transitions to {\em regular daily} and {\em reduced daily}: one gets guard $x_{dd}\le 9h$, the other guard $9h<x_{dd}\le 10h$ and action $c_{dd}:=c_{dd}+1$. Transitions from   {\em regular daily} and {\em reduced daily} get guard $c_{dd}\le 2$. This used constants $t_8:=10h=600$ and $t_9:=9h=540$.

\subsection{Weekly demands}

Article 4(i) defines a week as a calendar week, i.e., as the time between Monday 00:00 and Sunday 24:00.
Our formalization of real tachograph recordings by timed words replaces the  time-points of tachograph recordings by numbers starting from 0. Hence time is shifted and the information of the beginning of weeks is lost. A possibility to remedy this is to use 
timed  words where the beginnings of weeks are marked, or at least the first of them. For simplicity, we restrict attention to tachograph recordings starting at the beginning of a week, that is, we pretend that time-point 0 starts a week. We then leave it to the automaton to determine the time-points when weeks change. 

To this end, we use the auxiliary state {\em week} and a stopwatch $x_{\textit{week}}$ with bound $7\cdot 24h+1=168h+1$ that is active at all states except {\em accept} and {\em start}. 
All transitions to {\em accept} are guarded by $x_{\textit{week}}\le 168h$. The state {\em week} has incoming transitions from all  states except {\em accept} and transitions to all states except {\em start}. 
All these transitions are guarded by $x_{\textit{week}}= 168h$ and the outgoing transitions have actions $x_{\textit{week}}:=0$ and $c_{dd}:=0$ (see the implementation of  Article 6.1 above). 
This ensures that every accepting computation of $\A$ enters {\em week} for 0 time units exactly every week, i.e., every $168h$.

Additionally, we want the automaton to switch from {\em week} back  to the state it came from. To this end we introduce a bit $b_q$ for each state $q\neq\textit{accept}$. We give the transition from $q$ to {\em week} the action $b_q:=1$, and the transition from {\em week} to $q$ the guard $b_q=1$ and the  action $b_q:=0$.
The transition from {\em week} to {\em accept} has no guard involving the bits $b_q$. This uses the constant
$t_{10}:=168h=10080$.

Much of the following implementation work is done by adding guards and actions to the transitions from and to {\em week}. For example,
we can readily implement
\begin{quote}\sf
Article 6.2: The weekly driving time shall not exceed 56 hours and shall not result in the maximum weekly working time laid down in Directive 2002/15/EC being exceeded.

Article 6.3: The total accumulated driving time during any two consecutive weeks shall not exceed 90 hours.
\end{quote}

The time laid down by Directive 2002/15/EC is $60h$.
Use a stopwatch $x_{\textit{ww}}$ with bound  $60h+1$ that is active at {\em drive} and {\em other work}. Use a stopwatch  $x_{dw}$ with bound $56h+1$ active at {\em drive}.
To implement Article~6.2, the transitions   to {\em week} and {\em accept} have guard
$x_{dw}\le56h, x_{\textit{ww}}\le 60h$, and the transitions from {\em week} have action $x_{dw}:=0,x_{\textit{ww}}:=0$. Note that accepting computations contain only nodes with assignments satisfying $x_{dw}\le56h$ and $x_{\textit{ww}}\le 60h$. This implements Article 6.2.

To implement Article 6.3
we have to remember the value $x_{dw}$ of the previous week. We use a register~$x'_{dw}$ with the same bound as $x_{dw}$  and give the transitions from {\em week} the action $x'_{dw}:=x_{dw}$. Note $x'_{dw}$ functions like a register in that it just stores a value. We then guard all transitions to {\em accept} by $x'_{dw}+x_{dw}\le 90h$.
These articles use constants $ t_{11}:=56h=3360, t_{12}:=60h=3600$ and $ t_{13}:=90h=5400$.

\medskip

We now treat the articles concerning weekly rest periods. 
According to Article 4.(h), weekly rest periods can be regular or reduced, the former meaning at least $45h$ of rest, the latter means less than $45h$ but at least $24h$ of rest. These rest periods  are represented by the states {\em regular weekly} and {\em reduced weekly}. 

To implement their definition we use a stopwatch $x_{wr}$ with bound $45h$ active in these two states. For the two states we add transitions from and  to {\em drive} and {\em other work} and transitions to {\em accept}: those from {\em regular weekly} have guard $x_{wr}\ge 45h$ 
and action $x_{wr}:=0$, and those from {\em reduced weekly} have guards $45h>x_{wr}\ge 24h$ and action $x_{wr}:=0$. Later we shall refer to these guards as {\em definitorial guards} for {\em regular weekly} and {\em reduced weekly}, respectively. This uses the constants  $t_{14}:=45h=2700, t_{15}:=24h=1440$.

\medskip

\noindent We start with some easy implementations:
\begin{quote}\sf
Article 8.6 (3rd part): A weekly rest period shall start no later than at the end of six 24-hour periods from the end of the previous weekly rest period.

Article 8.3: A daily rest period may be extended to make a regular weekly rest period or a reduced weekly rest period.

Article 8.4: A driver may have at most three reduced daily rest periods between any two weekly rest periods.
\end{quote}

Article 8.6 (3rd part) is implemented with the help of a stopwatch $x_{\textit{pw}}$ that measures the time since the previous weekly rest period. It has bound $6\cdot 24h+1$ and
is active in all states excepot {\em start} and {\em accept}. We give the transitions to {\em regular weekly} and {\em reduced weekly} the guard $x_{\textit{pw}}\le 6\cdot 24h$, and the transitions from these two states the action~$x_{\textit{pw}}:=0$. This law uses constant $t_{16}:=6\cdot24h=8640$.

For Article 8.3 we simply copy the guards and actions of the transitions from {\em drive} and {\em other work} to {\em regular daily} to the corresponding transitions to both {\em regular weekly} and {\em reduced weekly}. Below we shall add more guards and actions. For Article 8.4 we use a counter $c_{rd}$ with bound $4$. 
We add guard $c_{rd}\le 2$
 and action $c_{rd}:=c_{rd}+1$ to the transitions to  {\em reduced daily} and the action $c_{rd}:=0$ to the transitions leaving {\em reduced weekly} and {\em regular weekly}.

\medskip

\noindent 
We still have to implement Article 8.1 for weekly rest periods, and additionally
\begin{quote}\sf
Article 8.9: A weekly rest period that falls in two weeks may be counted in either week, but not in both.
\end{quote}

We use two bits $b_{wr},b_{\textit{used}}$ meant to indicate whether a weekly rest period has been taken in the current week, and whether the current weekly rest period is used for this. 
The transitions from  {\em drive} or {\em other work} to
{\em reduced weekly} or {\em regular weekly} are duplicated: one gets guard $b_{wr}=0$ and action $b_{\textit{used}}:=1;\ b_{wr}:=1$, the other gets no further guards and actions. Transitions from {\em reduced weekly} or {\em regular weekly} get action $b_{\textit{used}}:=0$. The transitions to {\em week} get guard $b_{wr}=1$.

Each transition from {\em week} to {\em reduced weekly} or {\em regular weekly} is triplicated: the first gets additional guard $b_{\textit{used}}=1$ and action $b_{\textit{used}}:=0;\ b_{wr}:=0$, the second gets guard $b_{\textit{used}}=0$ and action $b_{wr}:=0$, and the third gets guard $b_{\textit{used}}=0$ and action $b_{\textit{used}}:=1;\ b_{wr}:=1$. 
This means that when the week changes during a weekly rest period and this rest period is not used, it can be used for the next week.

\medskip

\noindent The most complicated part of Regulation 561 are the rules governing reductions of weekly rest periods.  
The regulation starts as follows:
\begin{quote}\sf
Article 8.6 (1st part): In any two consecutive weeks a driver shall take at least
two regular weekly rest periods, or one regular weekly rest period and one reduced weekly rest period of at least 24 hours. 
\end{quote}

We use a bit $b_{rw}$ indicating whether the previous weekly rest period was reduced: transitions to {\em reduced weekly} have guard $b_{rw}=0$ and action $b_{rw}:=1$. Transitions to {\em regular weekly} have action $b_{rw}:=0$.
The regulation continues as follows:
\begin{quote}\sf
Article 8.6 (2nd part): However, the reduction shall be compensated by an equivalent period of rest taken en bloc before the end of the third week following the week in question.

Article 8.7:  Any rest taken as compensation for a reduced weekly rest period shall be attached to another rest period of at least nine hours.
\end{quote}

We introduce two registers $x_{c1},x_{c2}$ with bounds $45h-24h$.
We shall use the following informal mode of speech for the discussion: a reduced weekly rest period creates a `compensation obligation', namely an additional resting time $x_{c1}>0$ or $x_{c2}>1$. The obligations are `fulfilled' by setting these registers back to~0. Note that
compensation obligations are created by reduced weekly rest periods and, by Article 8.6 (1st part), this can happen at most every other week. As obligations have to be fulfilled  within 3 weeks, at any given time  a legal driver can have at most two obligations. 

We now give the implementation. Obligations are produced by transitions from {\em reduced weekly} (recall $x_{wr}$ records the resting time in {\em reduced weekly}): duplicate each such transition, give one guard $x_{c1}=0$ and action $x_{c1}:=45h-x_{wr}$, and the other guard $x_{c1}>0;\ x_{c2}=0$ and action $x_{c2}:=45h-x_{wr}$.
The 3 week deadline to fulfill the obligations is implemented by two counters $c_{c1},c_{c2}$ with bound 4. These counters are increased by transitions from {\em week} but only if some obligation is actually recorded: 
transitions from {\em week} get action $c_{c1}:= c_{c1}+\mathrm{sgn}(x_{c1});\ c_{c2}:= c_{c2}+\mathrm{sgn}(x_{c2})$. To ensure the deadline, transitions to {\em week} get guard $c_{c1}\le3;\ c_{c2}\le3$.

We now implement a way to fullfill obligations, i.e., to set $x_{c1}$ and $x_{c2}$ back to 0. This is done with the states {\em compensate1} and {\em compensate2} whose $\lambda$-label is $r$. We use a stopwatch~$x_{cr}$ with bound $45h-24h$ active at these states. We describe the transitions involving {\em compensate1}.  It receives transitions from the states with $\lambda$-label $r$, that is, {\em regular daily},
{\em reduced daily}, {\em regular weekly}, {\em reduced weekly} and {\em break}. The transition from
{\em break} has guard $x_{\textit{break}}\ge 9h$, the others have their respective definitorial guards (e.g., the one from {\em regular weekly} has guard $x_{wr}\ge 45h$). Transitions from {\em compensate1} go to 
{\em drive}, {\em other work} and {\em accept}. These have guard $x_{cr}\ge x_{c1}$ and action $x_{c1}:=0:\ c_{c1}:=0$. 
Additionally, we already introduced transitions from and to {\em week}: the transition to {\em week} is duplicated, one gets guard $x_{cr}<x_{c1}$, the other gets guard $x_{cr}\ge x_{c1}$ and action $x_{c1}:=0$. Thus, when the week changes during compensation and at a time-point when the obligation  is fulfilled, the counter is not increased.

The transitions from and to {\em compensate2} are analogous with $x_{c2},\ c_{c2}$ replacing $x_{c1},\ c_{c1}$.
These laws use the constant $t_{16}:=9h=540$; the bound $45h-24h=1560$ equals $t_{14}-t_{15}$.

\medskip

This finishes the definition of our automaton. We close this section with some remarks on the formalization:

\begin{remark}\
\begin{enumerate}\itemsep=0pt
\item[(a)] Regulation 561 contains a few laws concerning multi-manning that gives rise to an additional activity {\em available} and a distinction between breaks and rests. This is omitted in our treatment.

\item[(b)]  Article 7.2 is formally unclean: the second paragraph allows an exception to the first that  obviously cannot ``comply with the provisions of the first paragraph". A reasonable formalization requires an interpretational change to the law as written. The following two points or \cite{drive,Errezil:2019:Homologation} give  more such examples.

\item[(c)]  The definition in Article 4.(k) forgets the boundary case of a new driver:
without any (daily) rest period there cannot be any daily driving time. A similar problem appears with Article 8.6 (3rd part) when there is no previous weekly rest period.

\item[(d)] Concerning Article 6.1, recall that daily driving times are periods delimited by daily rest periods and a week is defined as calendar week starting at Monday 00:00. Consider a $10h$ extended daily driving time starting on a Sunday and ending on a Monday. To which one of the two weeks should it be counted? The law seems underspecified here. Our formalization 
assigns it to the week that starts on Monday. Various tachograph readers make different choices. For example, the software \textit{Police Controller} has an option to fix the choices or to choose the distribution as to minimize the fine~\cite{guillermo}.

\item[(e)] The nomenclature in Regulation 561 is confusing. A {\em day} is determined by daily rest periods, a {\em week} by the calendar, while  {\em weekly} (e.g., in Article 8.9) does not refer to calendar weeks. Additionally, the regulation does not state what should be done when a leap second is added on a Sunday so that the time 24:00:01 exists.

\item[(f)] For example, $(dr)^{270}$ is legal according to Article 7 but likely not  in line with the spirit of the law. Another regulation ((EU) 2016/799) stipulates that any minute of rest between two minutes of driving will be considered as driving -- outruling the above example. Then 
$(ddrr)^{135}$ is still legal. We expect that it is generally easy to construct artificial counterintuitive cases. 
\end{enumerate}
\end{remark}

\section{Theory of stopwatch automata}\label{sec:thy}

In this section we 
observe that stopwatch automata have the same expressive power as {\sf MSO} over finite words but a relatively tame model-checking complexity. We also give efficient algorithms for consistency-checking and scheduling (see Section~\ref{sec:formalmodel}). Finally, we mention a version of stopwatch automata going beyond {\sf MSO}.

\subsection{Expressivity}

\begin{lemma}\label{lem:regular} Every regular language is the language of some stopwatch automaton.
\end{lemma}

\begin{proof}
Given a non-deterministic finite automaton $\B=(S,\Sigma,I,F,\Gamma)$ as described in Subsection \ref{sec:buechi}, we define a stopwatch automaton~$\A = (Q, \Sigma, X, \lambda, \beta, \zeta, \Delta)$ such that $L(\B)=L(\A)$.

The states $Q$ of $\A$ are {\em start} and {\em accept} together with the states $S\times\Sigma$ labeled $\lambda \big(  (s,a) \big) :=a$ (the labels of  {\em start} and {\em accept} are irrelevant). 
We use a stopwatch $x$ intended to force the automaton to spend 1 time unit in every state $(s,a)$: it has
bound $\beta(x):=2$ and is active everywhere, i.e., $\zeta:=\{x\}\times Q$.

Each transition from some $(s,a)$ to $(s',a')$ has condition $x=1$ and action $x:=0$. We only allow those transitions from $(s,a)$ to $(s',a')$ when $(s,a,s') \in \Gamma$. This defines $\Delta$ when the states are not {\em start} or {\em accept}. Transitions from  {\em start}
lead to $I\times\Sigma$ and have guard $x=0$. Transitions to {\em accept} come from
$F\times \Sigma$ and have guard $x=0$.
\end{proof}

The converse of this lemma is based on the following definition.

\begin{definition}\label{defi:FAofSWA}
Given a
stopwatch automaton $\A=(Q,\Sigma, X, \lambda,\beta, \zeta,\Delta)$, we define the following finite  automaton $\B(\A)=(S,\Sigma,I,F,\Gamma)$. For $S$ we take the set of nodes of $\TS(\A)$; we let $I:=\{(\mathit{start},\xi_0)\}$ where  $\xi_0$ is constantly~0, and~$F$ contain the nodes $(q,\xi)$ of  $\TS(\A)$ such that $(q,\xi)\stackrel{0^*}{\to}(\textit{accept},\xi')$ for some assignment~$\xi'$. Here, $\stackrel{0^*}{\to}$ denotes the transitive and reflexive closure of $\stackrel{0}{\to}$. We let~$\Gamma$ contain $\big((q,\xi),a,(q',\xi')\big)$ if 
$\lambda(q')=a$ and there is $\xi''$ such that $(q,\xi)\stackrel{0^*}{\to}(q',\xi'')\stackrel{1}{\to}(q',\xi')$.
\end{definition}

\begin{theorem}\label{thm:regular} A language is regular if and only if it is the language of some stopwatch automaton.
\end{theorem}

\begin{proof} One direction follows from Lemma~\ref{lem:regular}. Conversely, given a language that is recognised by some SWA $\A$ we easily see that  $L(\A)=L(\B(\A))$ where $\B(\A)$ is as in Definition \ref{defi:FAofSWA}. That $L(\B(\A)) \subseteq L(\A)$ is immediate and for $L(\A)\subseteq L(\B(\A))$ we observe that a step in $\TS(\A)$ of duration $n$ can be obtained by $n$ consecutive steps of duration 1.
\end{proof}


The proof of  Lemma~\ref{lem:regular} gives a polynomial time computable function mapping every finite automaton to an equivalent stopwatch automaton. There is no such function for the converse translation, in fact, stopwatch automata are exponentially more succinct than finite automata.

\begin{proposition}\label{prop:succinct} For every $k$ there is a stopwatch automaton $\A_k$ of size $O(\log k)$ such that every finite automaton accepting $L(\A_k)$ has size at least $k$.
 \end{proposition}

We defer the proof to the end of  Section~\ref{sec:beyond}.

\subsection{Consistency-checking}

By Theorem \ref{thm:regular} we know that the languages accepted by stopwatch automata are exactly the regular languages. In particular, these languages are closed under intersections. We give an explicit construction of such an automaton computing an intersection because we shall need explicit bounds.

\begin{lemma} \label{lem:prod}
Given stopwatch automata $\A,\A'$ with bounds $B_\A,B_{\A'}$ one can compute in  time $O(\|\A\|\cdot\|\A'\|)$ a stopwatch automaton $\A\otimes\A'$ with bound $B_{\A\otimes\A'}=B_\A\cdot B_{\A'}$ such that $L(\A\otimes\A')=L(\A)\cap L(\A')$. 
\end{lemma}

\begin{proof} Let $\A=(Q,\Sigma, X, \lambda,\beta, \zeta,\Delta)$ and $\A'=(Q',\Sigma, X', \lambda',\beta', \zeta',\Delta')$ be stopwatch automata. Without loss of generality, we can assume that
\begin{enumerate}\itemsep=0pt
\item[(a)] $X$ and $ X'$ are disjoint;
\item[(b)] neither $\A$ nor $\A'$ contains transitions from its accept state;
\item[(c)] both $\Delta$ and $\Delta'$ contain for every state except the accept state a transition from the state to itself with trivial guard and action;
\end{enumerate}

We first define an automaton $
\A \times\A'$ with alphabet $\Sigma\times \Sigma$. Its states are $Q,\times Q'$ with start and accept state the pair of corresponding states of $\A$,~$\A'$. The stopwatches are $X\cup X'$ with the same bounds as in $\A,\A'$. A stopwatch $x\in X\cup X'$ is active in $(q,q')$ if either $(x,q)\in \zeta$ or $(x,q')\in \zeta'$.
A state $(q,q')$ is labeled $(\lambda(q),\lambda'(q'))$. The transitions are $((q_0,q'_0),g^*,\alpha^*,(q_1,q_1'))$
such that there are  $(q_0,g,\alpha,q_1)\in\Delta$ and $(q'_0,g',\alpha',q'_1)\in\Delta'$ such that $g^*$ computes the conjunction of $g$ and $g'$  and $\alpha^*$ executes $\alpha$ and $\alpha'$ in parallel. 

This is well-defined by (a). Also by (a) we can write assignments for $\A\times \A'$ as $\xi\cup \xi'$ where $\xi,\xi'$ are assignments for $\A,\A'$. We claim that $\A\times \A'$ accepts a word $(a_0,a'_0)\cdots (a_{n-1},a'_{n-1})\in(\Sigma\times\Sigma)^n$ if and only if $a_0\cdots a_{n-1}\in L(\A)$ and  $a'_0\cdots a'_{n-1}\in L(\A')$. 

 Indeed, if $((q_0,q_0'),\xi_0\cup \xi_0')\stackrel{t_0}{\to}\cdots((q_{\ell-1},q'_{\ell-1}),\xi_{\ell-1}\cup \xi_{\ell-1}')$ is an initial accepting run of $\A\times \A'$, then, by (b), $q_i$ is the accept state of $\A$ exactly for $i=\ell-1$. Then $(q_0,\xi_0)\stackrel{t_0}{\to}\cdots (q_{\ell-1},\xi_{\ell-1})$ is an initial accepting run of $\A$ that reads $a_0\cdots a_{n-1}$. Analogously, $a'_0\cdots a'_{n-1}\in L(\A')$. 
 
Conversely, given $a_0\cdots a_{n-1}\in L(\A)$ and  $a'_0\cdots a'_{n-1}\in L(\A')$ we can choose
initial accepting runs of $\A$ and $\A'$ reading these words, respectively, and have  the form:
 \begin{eqnarray*}
 & (q_{0},\xi_{0})\stackrel{0^*}{\to}(r_{0},\eta_0)\stackrel{1}{\to}  (q_{1},\xi_{1})\stackrel{0^*}{\to}(r_{1},\eta_1)\stackrel{1}{\to}\cdots (q_{\ell-1},\xi_{\ell-1})\stackrel{0^*}{\to}(r_{\ell},\eta_\ell),\\
 & (q'_{0},\xi'_{0})\stackrel{0^*}{\to}(r'_{0},\eta'_0)\stackrel{1}{\to}  (q'_{1},\xi'_{1})\stackrel{0^*}{\to}(r'_{1},\eta'_1)\stackrel{1}{\to}\cdots (q'_{\ell'-1},\xi'_{\ell'-1})\stackrel{0^*}{\to}(r'_{\ell'},\eta'_{\ell'}).
\end{eqnarray*}
Here, $\stackrel{0^*}{\to}$ denotes the transitive closure of $\stackrel{0}{\to}$ in $\TS(\A)$ and $\TS(\A')$. Then $\ell=\ell'=n$. By~(c), we can assume that the $\stackrel{0^*}{\to}$-paths have the same length. Then the runs have the same length. Then the runs can be combined in the obvious way to an initial accepting run of $\A\times\A'$ reading $(a_0,a'_0)\cdots(a_{n-1},a'_{n-1})$. This proves the claim.

The automaton $\A\otimes\A'$ is easily obtained from a modification of $\A\times\A'$ whose initial accepting runs are precisely those initial accepting runs of $\A\times\A'$  that read words
over $\{(a,a)\mid a\in\Sigma\}$. Such a modification is easy to obtain: add a new stopwatch $y$ with bound $1$ to $\A\times\A'$ that is active in all states; every transition gets action $y:=0$ and every transition from a state $(q,q')$ with $\lambda(q)\neq\lambda(q')$ gets guard $y=0$.

The claims about the bound of $\A\otimes\A'$ and the time needed to compute it are clear.
\end{proof}

The following algorithm can be used to check if the intersection of two languages is empty or not. Informally, we can perceive this as an algorithm that checks whether a certain type of behaviour  is illegal according to a law when both the type of behaviour and the law are specified by  stopwatch automata. 

\begin{theorem}\label{thm:conscheck}
There is an algorithm that given stopwatch automata $\A,\A'$ with bounds $B_\A,B_{\A'}$, respectively,  decides whether  $L(\A)\cap L(\A')\neq\emptyset$ in time $$O\big((\|\A\|\cdot\|\A'\|\cdot B_\A\cdot B_{\A'})^3\big).$$
\end{theorem}

\begin{proof} 
\moritz{Here and in the next theorem, I understand the algorithm, but can you please comment on the exponents? WHy three and two. I would like to see some hints here.}
The algorithm first computes the product automaton $\A\otimes\A'$ from the previous lemma. Next, the algorithm
 computes the finite automaton $\B(\A\otimes\A')=(S,\Sigma,I,F,\Gamma)$ as given in Definition \ref{defi:FAofSWA}. Note $|S|\le O(\|\A\| \cdot\| \A\|\cdot B_{\A\otimes\A'})$. 
 
 To compute $\Gamma$ we first compute  the graph on $S$ with edges~$\stackrel{0}{\to}$: cycle through all $(q,\xi)\in S$ and transitions $\Delta$ of $\A\otimes\A'$ and evaluate its guard and action on~$\xi$.  Each evaluation can be done in time linear in the size of the circuits, so in time $O(\|\A\|\cdot\|\A'\|)$. Thus, the graph can be computed in time  $O(|S|\cdot|\Delta|\cdot\|\A\|\cdot\|\A'\|)$. Its transitive closure can be computed in cubic time $O(|S|^3)$. 
 Each of the at most $|S|^2$ edges in $\stackrel{0^*}{\to}$ determines a transition in $\Gamma$.
Thus $\B(\A\otimes\A')$ can be computed in time cubic in $\|\A\|\cdot\|\A'\|\cdot B_\A\cdot B_{\A'}$.

Observe $L(\B(\A\otimes\A'))\neq \emptyset$ if and only if some final state is reachable from the initial state. Checking this takes linear time in the size of the automaton. 
\end{proof}

The algorithm solves the consistency problem for stopwatch automata by fixing input $\A'$ to some stopwatch automaton  with $L(\A')=\Sigma^*$.

\begin{corollary}\label{cor:conscheck}
There is an algorithm that given a  stopwatch automaton $\A$ with bound~$B_\A$, decides whether  $L(\A)\neq\emptyset$ in time $$O\big(\|\A\|^3\cdot B^3_{\A}\big).$$
\end{corollary}

\subsection{Model-checking}


The algorithm of Theorem~\ref{thm:conscheck} can be used to solve the model-checking problem: 
note $w\in L(\A)$ if and only if $L(\B_w)\cap L(\A)\neq \emptyset$ for a suitable size $O(|w|)$ automaton $\B_w$ with $L(\B_w)=\{w\}$.
A more direct model-checking algorithm achieves a somewhat better time complexity, in particular,  linear in $B_\A$:

\begin{theorem}\label{thm:modcheck}
There is an algorithm that given a word $w$  and a  stopwatch automaton~$\A$ with bound $B_\A$  decides whether  $w\in L(\A)$
 in time $$
 O\big(\|\A\|^2 \cdot B_\A \cdot  |w|\big).$$
\end{theorem}

\begin{proof}  Let $ \A=(Q,\Sigma, X, \lambda,\beta, \zeta,\Delta)$ have bound $B_\A$.
Let $G=(V,E)$ be the directed graph whose vertices $V$ are the nodes of $\TS(\A)$ and whose directed edges $E$ are given by $\stackrel{0}{\to}$. Note $|V|=|Q|\cdot B_\A$ and $|E|\le B_\A\cdot |\Delta|$. Let $w=w_1\cdots w_t\in\Sigma^t$ for some $t\in\N$.

We define a  directed graph with vertices $\{0,\ldots,t \}\times V$ and the following edges. Edges within each copy $\{i\}\times V$ are copies of $E$. So these account for at most $(t+1)\cdot B_\A\cdot |\Delta|$ many edges, each  determined by evaluating guards and actions in time $O(\|\A\|)$.
Further  edges lead from vertices in the $i$-th copy $\{i\}\times V$ to vertices in the $(i+1)$th copy $\{i+1\}\times V$, namely from $(i,(q,\xi))$ to $(i+1,(q,\xi'))$ if 
\begin{eqnarray}\label{eq:modcheck}
&&\text{$(q,\xi)\stackrel{1}{\to}(q,\xi')$ and $q\neq\mathit{accept}$ and
$\lambda(q)=w_i$.}
\end{eqnarray}
There are at most $t\cdot|Q|\cdot B_\A$ such edges between copies. This graph has size
$O(t\cdot\|\A\|\cdot B_\A)$ and can be computed in time $O(t\cdot\|\A\|^2\cdot B_\A)$.

It is clear that $w\in L(\A)$ if and only if
$(t,(\textit{accept},\xi'))$ for some assignment $\xi'$ is {\em reachable} in the sense that
there is a  path from $(0,(\mathit{start},\xi_0))$ with $\xi_0$ constantly 0  to it. 
Checking this takes time linear in the size of the graph.
\end{proof}

\subsection{Scheduling}

We  strengthen the model-checker of Theorem~\ref{thm:modcheck} to solve the scheduling problem: the model-checker treats the special case for inputs with $n=0$.

\begin{theorem}\label{thm:scheduling}
There is an algorithm that given a stopwatch automaton $\A$ with bound $B_\A$ and alphabet~$\Sigma$, a  word $w\in\Sigma^*$, a letter $a\in\Sigma$ and $n\in\N$, rejects if there does not exist a word $v$ over $\Sigma$ of length~$n$ such that  $wv\in L(\A)$ and otherwise computes such a word $v$ with maximal $\#_a(v)$. It runs in time $$ O\big( \|\A\|^2 \cdot B_\A \cdot(|w|+n)\big).$$
\end{theorem}

\begin{proof} Consider the graph constructed in the poof of Theorem~\ref{thm:modcheck} but with $t+n$ instead of $t$ and the following modification: in \eqref{eq:modcheck}
for $t\le i< t+n$ drop the condition $\lambda(q)=w_i$ for edges between the $i$-th and the $(i+1)$-th copy.
In the resulting graph there is a reachable vertex $(t+n,(\textit{accept},\xi'))$ for some assignment $\xi'$ if and only if there exists a length $n$ word $v$ such that $wv\in L(\A)$. We now show how to compute the maximum value $\#_a(v)$ for such  $v$.

Successively for $i=0,\ldots,n$ compute a label $V_i(q,\xi)$ for each vertex $(t+i,(q,\xi))$ in the $(t+i)$-th copy. For $i=0$ all these labels are $\#_a(w)$. For $i>0$ label $(t+i,(q,\xi))$ with the maximum value
$$
\left\{\begin{array}{ll}
V_{i-1}(q,\xi')+1&\text{if }\lambda(q)=a,\\
V_{i-1}(q,\xi')&\text{else}
\end{array}
\right.
$$
taken over $\xi'$ such that there is an edge from $(t+i-1,(q,\xi'))$ to $(t+i,(q,\xi))$. Then the desired maximum value $\#_a(v)$ is the maximum label $V_n(q,\xi)$ such that $q=\mathit{accept}$ and
$(t+n,(q,\xi))$ is reachable.

Additionally we are asked to compute a word $v$ witnessing this value. To do so the labeling algorithm 
computes a set of directed edges, namely for each $(t+i,(q,\xi))$ with $i>0$ to a vertex $(t+i-1,(q,\xi'))$ witnessing the maximum value above. This set of edges defines a partial function that, for each $i>0$, maps vertices in the $(t+i)$-th copy to vertices in the $(t+i-1)$-th copy. To compute $v$ as desired start at a vertex $(t+n,(q_n,\xi_n))$ witnessing the maximal value $\#_a(v)$ and iterate this partial function to get a sequence
of vertices $(t+i,(q_i,\xi_i))$. Then $v:=\lambda(q_1)\cdots\lambda(q_n)$ is as desired.

It is clear that all this can be done in time linear in the size of the graph.
\end{proof}

\subsection{Beyond regularity}\label{sec:beyond}

A straightforward generalization of stopwatch automata allows $\beta$ to take  value~$\infty$. An {\em unbounded stopwatch automaton} is a stopwatch automaton where $\beta$ is the function constantly $\infty$. 
We note that model-checking is undecidable already for {\em simple} such automata (see \cite[Proposition~1]{update} for a similar proof). These simple automata use two stopwatches $x,y$ that are nowhere active (i.e., $\zeta=\emptyset$), all guards check $z=0$ or $z\neq 0$, and all actions are either $z:=z+1$ or $z:=z\dot-1=\max\{z-1,0\}$ for some $z\in\{x,y\}$. 

\begin{proposition} There is no algorithm that given a simple unbounded stopwatch automaton decides whether it accepts the empty word.
\end{proposition}

\begin{proof} Recall, a {\em two counter machine} operates two variables $x,y$ called {\em counters} and is given by a finite non-empty sequence $(\pi_0,\ldots,\pi_{\ell}) $ of {\em instructions} $\pi_i$, namely, either $z:=z+1,z:=z\dot- 1$, ``Halt'' or ``if $z=0$, then goto $j$, else goto $k$'' where $z\in\{x,y\}$ and $j,k\le\ell$; exactly $\pi_{\ell}$ is ``Halt''. The computation (without input) of the machine is straightforwardly explained. It is long known that it is undecidable whether a given two counter machine halts or not. 

Given such a machine $(\pi_0,\ldots,\pi_{\ell})$ it is easy to construct a simple automaton that accepts the empty word if and only if the two counter machine halts. It has states $Q=\{0,1,\ldots,\ell\}$ understanding $\mathit{start}=0$ and $\ell=\mathit{accept}$; $\Sigma$ and $\lambda$
are unimportant, and $\Delta$ is defined as follows. If $\pi_i$ is the instruction $z:=z+1$, then
add the edge $(i,g,\alpha,i+1)$ where $g$ is trivial and $\alpha$ changes $z$ to $z+1$.  If $\pi_i$ is the instruction $z:=z\dot-1$, proceed similarly. If $\pi_i$ is ``if $z=0$, then goto $j$, else goto $k$'' add edges $(i,g,\alpha,j),(i,g',\alpha,k)$
where~$g$ checks $z=0$ and $g'$ checks $z\neq 0$ and $\alpha$ computes the identity. 
\end{proof}

What seems to be a middle ground between unbounded stopwatches and stopwatches with a constant bound is to let the bound grow with the length of the input word.

The definition of a stopwatch automaton $\A=(Q,\Sigma,X,\lambda,\beta,\zeta,\Delta)$ can be generalized letting $\beta:X\times \N\to\N$ be {\em monotone} in the sense that
$\beta(x,n)\le \beta(x,n')$ for all $x\in X,\  n,n'\in\N$ with $n\le n'$. We call this a {\em $\beta$-bounded stopwatch automaton} and call $B_\A:\N\to\N$ defined by
$$
B_\A(n):=\prod_{x\in X}(\beta(x,n) +1)
$$ 
the {\em bound of $\A$}.
For each $n\in\N$ we have a stopwatch automaton $\A(n):=(Q,\Sigma,X,\beta_n,\zeta,\lambda,\Delta)$ where $\beta_n:X\to\N$ maps $x\in X$ to $\beta(x,n)$; note $B_{\A(n)}=B_\A(n)$.

The {\em language $L(\A)$ accepted} by a $\beta$-bounded stopwatch automaton $\A$ contains a word $w$
over $\Sigma$ if and only if $w\in L(\A(|w|))$.

\begin{proposition} A  language is accepted by some stopwatch automaton  if and only if it is accepted by some $\beta$-bounded stopwatch automaton with bounded $\beta$.
\end{proposition}

\begin{proof} Let $\A$ be a $\beta$-bounded stopwatch automaton for bounded $\beta$. There is $n_0\in\N$ such that $\beta(x,n)=\beta(x,n_0)$ for all $x\in X$ and $n\ge n_0$. Hence $L(\A(n_0))$ and $L(\A)$ contain the same words of length at least $n_0$. Since there are only finitely many shorter words, and $L(\A(n_0))$ is regular by Theorem~\ref{thm:regular}, also $L(\A)$ is regular. 
\end{proof}

Theorem~\ref{thm:modcheck} on feasible model checking generalizes:

\begin{corollary} Let $X$ be a finite set and assume $\beta:X\times\N\to\N$ is such that $\beta(x,n)$ is computable
from $(x,n)\in X\times \N$ in time $O(n)$.  
Then there is an algorithm that given a word $w$  and a $\beta$-bounded stopwatch automaton $\A$ with bound $B_\A:\N\to\N$ decides whether  $w\in L(\A)$
 in time $$
 O\big(\|\A\|^2 \cdot B_\A(|w|) \cdot  |w|\big).$$
\end{corollary}

If $\beta(x,n)$ grows slowly in $n$ this can be considered tractable. Any growth, no matter how slow, leads to non-regularity:

\begin{proposition} Let $f:\N\to\N$ be unbounded and non-decreasing. Then there is a $\beta$-bounded stopwatch automaton $\A=(Q,\Sigma,X,\lambda, \beta,\zeta,\Delta)$ with $\beta(x,n)= f(n)$ for all $x\in X$ and all $n\in \N$  such that $L(\A)$ is not regular.
\end{proposition}

\begin{proof} Let $\Sigma$ be the three letter alphabet $\{a,b,c\}$, and let $L$ contain a length $t$ word over~$\Sigma$ if it has the form $a^{s}b^sc^*$ for some $s< f(t)$. 
Since $f$ is unbounded, $L$ contains such words for arbitrarily large $s$. It thus follows from the Pumping Lemma, that $L$ is not regular.

It suffices to define a $\beta$-bounded stopwatch automaton $\A$ such that 
that accepts a word of sufficiently large length $t$ if and only if it belong to $L$. 
The  states are $\mathit{start},\mathit{accept}, q_a,q_b,q_c$ with $\lambda$-labels $a,a,a,b,c$, respectively.
We use stopwatches $x_a,y_a,x_b$ all with bound $f(t)$ and 
declare $x_a,y_a$ active in $q_a$ and {\em start}, and $x_b$ active in $q_b$.
There are transitions from $\mathit{start}$ to~$q_a$, from $q_a$ to $q_b$, from $q_b$ to $q_c$, and from $q_c$ to {\em accept} -- described next.

The transition from $\mathit{start}$ to $q_a$ has guard $x_a=0$ and action $y_a:=1$. 
For sufficiently large~$t$, the bound $f(t)$ of $x_a$ is positive. Then
any initial accepting computation (of $\A$ on a word of length $t$)  spends 0 time in $\mathit{start}$, 
and thus starts $(\mathit{start},[0,0,0])\stackrel{0}{\to}(q_a,[0,1,0])$; we use a notation like $[1,2,3]$ to denote the assignment that maps $x_a$ to $1$, $y_a$ to $2$, and $x_b$ to 3. 

The transition from $q_a$ to $q_b$ has guard $x<y$ and trivial action. An initial accepting computation on a word of length $t$ 
can stay in $q_a$ for some time $r$ reaching $(q_a,[r,r+1,0])$ 
for $r<f(t)$, or reaching  $[f(t),f(t),0]$ for $r\ge f(t)$ due to the bound of $x,y$. In the latter case  the transition to $q_b$ is disabled and {\em accept} cannot be reached. Staying  in $q_a$ for any time $s< f(t)$ allows the transition to $q_b$.

The transition from $q_b$ to $q_c$ has guard $x_a=x_b$ and trivial action. The transition from $q_c$ to {\em accept} has trivial guard and action.
\end{proof}

We can now prove that stopwatch automata are exponentially more succinct than finite automata as was expressed in Proposition~\ref{prop:succinct}.

\begin{proof}[Proof of Proposition~\ref{prop:succinct}.] Consider the previous proof for the function $f$ constantly $k$. Clearly, then $L$ is regular. By the Pumping Lemma, a finite automaton accepting $L$ has at least $k$ states. The stopwatch automaton $\A$ accepts $L$ and has size $O(\log k)$. Indeed, the size of a binary encoding of $\A$ is dominated by the bits required to write down the bound $k$ of the stopwatches.
\end{proof}

\section{Discussion and a lower bound}\label{sec:concl}

We suggest the model-checking problem for stopwatch automata  and finite words (over some finite alphabet) as an answer to our central question in Section~\ref{sec:formalmodel}, the quest for a model for algorithmic laws concerning activity sequences. 
This section discusses to what extent this model meets the three desiderata listed in Section~\ref{sec:formalmodel}, and mentions some open ends for future work.

\subsection{Summary}

\paragraph{Expressivity} Stopwatch automata are highly expressive, namely, by Theorems~\ref{thm:regular} and~\ref{thm:buechi}, equally expressive as $\mathsf{MSO}$ (over finite words). In particular, \cite{drive} argued that Regulation 561 is expressible in $\mathsf{MSO}$, so it is also expressible by stopwatch automata. In Section~\ref{sec:beyond} we showed that a straightforward generalization of stopwatch automata can go even beyond $\mathsf{MSO}$. Future research might show whether this is useful for modeling actual laws. 

\begin{example} Imagine an employee who can freely schedule his work and choose among 
various activities $\Sigma$ to execute at any given time point. The employer favors an activity $a\in\Sigma$ and checks at random time-points that the employee used at least a third of his work-time on activity $a$ since the previous check. 
The set of $w\in\Sigma^*$ with 
$\#_a(w)\ge |w|/3$ is not regular but is accepted by a simple $\beta$-bounded stopwatch automaton with one stopwatch $x$ and bound $\beta(x,t)=\lceil t/3\rceil$.
\end{example}

\paragraph{Naturality} We stressed that expressivity alone is not sufficient, {\em natural} expressivity is required. This is an informal requirement, roughly, it means that the specification of a law should be readable, and in particular, not too large. In particular, as emphasized in Section~\ref{sec:buechi}, constants appearing in laws bounding durations of certain activities should not blow up the size of the formalization (like it is the case for $\LTL$). We suggest that our expression of Regulation 561 by a stopwatch automaton is natural. 

There is a possibility to use stopwatch automata as a {\em law maker}:
an interface that allows to specify laws in a formally rigorous way without assuming much mathematical education. It is envisionable to use graphical interfaces akin to the one provided by UPPAAL\footnote{\url{https://uppaal.org/}} to draw stopwatch automata. A discussion of this possibility as well as  the concept of ``readability'' is outside the scope of this paper.

\paragraph{Tractability} The main constraint of a model-checking problem as a formal model for algorithmic law is its computational tractability.
In particular, the complexity of this problem should scale well with the constants appearing in the law.  This asks for a fine-grained complexity analysis taking into account various aspects
of a typical input, and, technically, calls for a complexity analysis  in the framework of parameterized complexity theory.
Theorem~\ref{thm:modcheck} gives a model-checker for stopwatch automata. Its worst case time complexity upper bound scales transparently with the involved constants, and, most importantly, the runtime is not exponential in these constants. This overcomes a bottleneck of many model-checkers designed in the context of system verification (see Section~\ref{sec:timedmodal}). Theorems~\ref{thm:conscheck} and \ref{thm:scheduling} give similar algorithms for consistency-checking and scheduling.

\subsection{Parameterized model-checking}\label{sec:pswa}

We have an upper bound $O(\|\A\|^2\cdot B_\A\cdot |w|)$ to the worst case runtime of our model-checker. The troubling factor is $B_\A$: the runtime grows fast with the stopwatch bounds of the automaton.
Intuitively, these bounds stem from the constants mentioned by the law as duration constraints on activities. 
At least, this is the case for our formalization  of Regulation~561: we explicitly mentioned 17 constants $\bar t=(t_0,\ldots, t_{16})$ 
which determine our automaton, specifically its  bounds, guards and actions.
To wit,  $\bar t$ determines bounds on stopwatches as follows:
$$
 \begin{array}{|c|c|c|c|c|c|c|c|c|c|c|c|}
  x_{\textit{break}}&x_{cd} &x_{\mathit{day}} & x_{dr} & x_{dd} &x_{\mathit{week}}&x_{\textit{ww}}&x_{dw} ,x'_{dw}&x_{wr}&x_{\textit{pw}}&x_{c1},x_{c2}\\\hline
t_{16}&t_0+1&t_3+1&t_4                           &t_8+1    &t_{10}+1             &t_{12}+1&t_{11}+1     &t_{14} &t_{16}+1&t_{14}-t_{15}
\end{array}
$$
The other stopwatches have bounds independent of $\bar t\in\N^{17}$. For any choice of $\bar t$ we 
get an automaton $\A(\bar t)$ that accepts exactly the words that represent  activity sequences 
that are legal according  to the variant of Regulations 561 obtained by changing these constants to~$\bar t$. It is a matter of no concern to us that not all choices for~$\bar t$ lead to meaningful laws. We are interested in how the runtime 
of our model-checker for Regulation 561 depends on these constants. 
By Theorem~\ref{thm:modcheck} we obtain:
 
\begin{corollary} There is an algorithm that given $\bar t\in\N^{17}$ and a word $w$  decides whether  $w\in L(\A(\bar t))$ in time 
$$
O\big(t^2_{16}\cdot t_0\cdot t_3\cdot t_4\cdot t_8\cdot t_{10}\cdot t_{12}\cdot t_{11}^2\cdot t_{14}\cdot(t_{14}-t_{15})^2\ \cdot \ |w|\big).
$$
\end{corollary}

For the actual values of $\bar t$ in Regulation 561 the above product of the $t_i$'s evaluates to the number
$$
6006978697267786744332288000000000000000.
$$
This casts doubts whether the factor $B_\A$ in our worst-case runtime $O(\|\A\|^2\cdot B_\A\cdot |w|)$ should be regarded tractable. 
Can we somehow improve the runtime dependence from the constants? 

For the sake of discussion, note that $B_\A$ is trivially bounded by $t_\A^{c_\A}$ where $c_\A$ is the number of stopwatches of $\A$ and $t_\A$ is the largest bound of some stopwatch of~$\A$ (as in Section~\ref{sec:techsum}). Intuitively, $c_\A$ is ``small'' but $t_\A$ is not. In the spirit of parameterized complexity theory it is natural to ask  whether the factor ($B_\A$ or) $t_\A^{c_\A}$ can be replaced by  $f(c_\A)\cdot t_\A^{O(1)}$ for some computable function $f:\N\to\N$. We now formulate this question precisely
in the framework of parameterized complexity theory.

The canonical parameterized version of our model-checking problem is

\medskip
 
\fbox{\begin{tabular}{ll}
{\em Input:}& a stopwatch automaton $\A=(Q,\Sigma, X, \lambda,\beta, \zeta,\Delta)$ and  $w\in\Sigma^*$.\\
{\em Parameter:} &$\|\A\|$.\\
{\em Problem:}& $w\in L(\A)$ ?
\end{tabular}}
\medskip

\noindent Our model-checker of Theorem~\ref{thm:modcheck} 
witnesses that this problem is fixed-parameter tracta\-ble. Indeed, $\|\A\|^2\cdot B_\A\le f(\|\A\|)$
for some computable $f:\N\to\N$ because the circuits in $\A$ have size $\ge \log B_\A$.
 Intuitively, that $B_\A$ is bounded in terms of the parameter $\|\A\|$ means that the parameterized problem above models instances where $B_\A$ is ``small'', in particular $\beta$ takes ``small'' values. But there are cases of interest where this is not true: the constant $t_{10}:=10080$ in Regulation 561 is not ``small''.
In the situation of such an algorithmic law, the above parameterized problem is the wrong model. 

A better model parameterizes a model-checking instance $(\A,w)$ by the size of~$\A$ but discounts the stopwatch bounds. More precisely, consider the following parameterized problem:

\medskip
 
\fbox{\begin{tabular}{ll}
$p$-\textsc{SWA}&\\
{\em Input:}& a stopwatch automaton $\A=(Q,\Sigma, X, \lambda,\beta, \zeta,\Delta)$ and  $w\in\Sigma^*$.\\
{\em Parameter:} &$|Q|+|\Sigma|+|X|+|\Delta|$.\\
{\em Problem:}& $w\in L(\A)$ ?
\end{tabular}}
\medskip

 Note that the algorithm of
Theorem~\ref{thm:modcheck} does not witness that this problem would be fixed-parameter tractable.
We arrive at the precise question: 
\begin{center}\em
Is $p$-\textsc{SWA} fixed-parameter tractable?  
\end{center}

\subsection{A lower bound}\label{sec:lower}

In this section we prove that the answer to the above question is likely negative:

\begin{theorem}\label{thm:pswa} $p$-\textsc{SWA} is not fixed-parameter tractable unless every problem in the  W-hierarchy is fixed-parameter tractable. 
\end{theorem}

We refer to any of the monographs \cite{df,fg,df2} for a definition of the {\em W-hierarchy} $\W[1]\subseteq \W[2]\subseteq \cdots$. As mentioned in Section \ref{sec:formalmodel}, the central hardness hypothesis of parameterized complexity theory is that already the first level
$\W[1]$ contains problems that are not fixed-parameter tractable.
We thus consider Theorem~\ref{thm:pswa} as strong evidence that the answer to our question is negative.

We prove Theorem~\ref{thm:pswa} by a reduction from a parameterized version of the
 {\em Longest Common Subsequence Problem (LCS)}. This classical problem takes as inputs an alphabet $\Sigma$, finitely many words $w_0,\ldots,w_{k-1}$ over $\Sigma$ and a natural number $m$. The problem is to decide whether the given words have a common subsequence of length $m$: such a subsequence is a length $m$ word $a_0\cdots a_{m-1}$ over $\Sigma$ (the $a_i$ are letters from $\Sigma$) that can be obtained from every $w_i,i<k,$ by deleting some letters. In other words,
 for every $i<k$ there are $ j^i_{0}<\cdots<j^i_{m-1}<|w_i|$ such that for all $\ell<m$ 
the word $w_i$ has letter $a_\ell$ at position $j^i_{\ell}$. For example, both $bbaccb$ and $bbaacb$ are common subsequences of $abbaaccb$ and $bbacccacbb$.

This problem received considerable attention in the literature and has several natural  parameterized versions \cite{lcs1conf,lcs1,lcs2,lcs3}. 
We consider the following one:\footnote{In \cite{fg} the notation $p$-\textsc{LCS} refers to a different parameterization of LCS.}

\medskip
 
\fbox{\begin{tabular}{ll}
$p$-\textsc{LCS}&\\
{\em Input:}& an alphabet $\Sigma$, words $w_0,\ldots,w_{k-1}\in\Sigma^*$ for some $k\in\N$, and $m\in\N$. \\
{\em Parameter:} &$k+|\Sigma|$.\\
{\em Problem:}& do $w_0,\ldots,w_{k-1}$ have a common subsequence of length $m$ ?
\end{tabular}}

\medskip

The statement that $p$-\textsc{LCS} is fixed-parameter tractable means that it can be decided by an algorithm that on an instance $(\Sigma,w_0,\ldots,w_{k-1},m)$ runs in time
$$f(k+|\Sigma|)\cdot (|w_0|+\cdots +|w_{k-1}|)^{O(1)} $$ for some computable function $f:\N\to\N$.
The existence of such an algorithm is unlikely due to
the following result:

\begin{theorem}[\cite{lcs2}] \label{thm:lcs} $p$-\textsc{LCS} is not fixed-parameter tractable unless every problem in the  W-hierarchy is fixed-parameter tractable. 
\end{theorem}


\begin{proof}[Proof of Theorem~\ref{thm:pswa}:]
Let $(\Sigma, w_0,\ldots,w_{k-1},m)$ be an instance of $p$-\textsc{LCS}, so $\Sigma$ is an alphabet, $w_0,\ldots,w_{k-1}\in\Sigma^*$ and $m\in\N$. 
Let $w:=w_0\cdots w_{k-1}$ be the concatenation of the given words, and consider $w^m$, the concatenation of $m$ copies of $w$.
We construct a $P(\Sigma)$-labeled stopwatch automaton $\A=(Q,\Sigma, X, \lambda,\beta, \zeta,\Delta)$ that accepts~$w^m$ if and only if $w_0,\ldots,w_{k-1}$ have a common subsequence of length $m$. 

An initial accepting computation of $\A$ on $w^m$ proceeds in $m$ {\em rounds}, each round reads a copy of $w$.  In  round $\ell<m$ the computation guesses a position within each of the words $w_0,\ldots , w_{k-1}$ copied within $w$, and ensures they all carry the same letter. These positions are stored in registers (i.e., nowhere active stopwatches) $x_0,\ldots,x_{k-1}$ with bounds $|w_0|+1,\ldots,|w_{k-1}|+1$, respectively. Our intention is that the value of $x_i$ after round $\ell<m$ equals the position  $j^i_{\ell}$ in the definition of a common subsequence.

Our intention is that an initial accepting computation in round $\ell<m$ cycles though~$k$ many {\em guess parts} of the automaton. Within guess part 0, the computation reads $w_0$ (within copy $\ell$ of $w$ in $w^m$), within guess part 1 the computation reads~$w_1$ and so on.
The states of $\A$ are the states of the guess parts plus an an additional state {\em accept}.
Each guess part consists of a copy of the states {\em start}, {\em end}, and  {\em guess$(a)$} for $a\in\Sigma$. The $\lambda$-labels of {\em start} and {\em end} are $\Sigma$, the $\lambda$-label of {\em guess$(a)$} is $\{a\}$. The start state of $\A$ is {\em start} in guess part 0.

We intend  that the computation in guess part $i<k$ spends some time $t<|w_i|$ in {\em start}, then spends  exactly one time unit in some state {\em guess$(a)$}, and then spends time $|w_i|-t$ in {\em end} before switching to the next guess part. The position guessed is $t$ and stored as the value of~$x_i$. Writing momentarily $w_i= a_0 a_1\cdots a_{|w_i|-1}$ the computation reads the (possibly empty) word $a_0\cdots a_{t-1}$ in state {\em start}, then reads $a_t$ in state {\em guess$(a_t)$}, and then reads the (possibly empty) word $a_{t+1}\cdots a_{|w_i|-1}$ in state {\em end}.

We enforce this behavior as follows. 
There are transitions from {\em start} (in guess part~$i$) to {\em guess$(a)$} for every $a\in\Sigma$, and for every $a\in\Sigma$ from {\em guess$(a)$} to {\em end}.  We use a stopwatch~$y_i$ with bound $|w_i|+1$ active in all states of guess part $i$ and a stopwatch~$z$ with bound $2$ active in the states  {\em guess$(a)$}, $a\in\Sigma$, of any  guess part.  It will be clear that  initial accepting computations enter guess part $i$ with both $y_i$ and $z$ having value~0.
The transitions from {\em start} to {\em guess$(a)$}, $a\in\Sigma$, have guard checking 
$x_i<y_i< |w_i|$ and action setting $x_i:=y_i$. The transitions from {\em guess$(a)$}, $a\in\Sigma$, to {\em end} have  guard checking $z=1$ and  action setting $z:=0$. The state {\em end} in guess part $i<k-1$ has a transition to {\em start} in guess part $i+1$; for $i=k-1$ this transition is to {\em start} in guess part~0. These transitions have guard checking $y_i=|w_i|$ and action setting $y_i:=0$.

Observe that the computation spends time $|w_i|$ in guess part $i<k$ and increases the value of $x_i$. Hence the values of $x_i$ after each round form an increasing sequence of positions $<|w_i|$. We have to ensure that the values of $x_0,\ldots, x_{k-1}$ after a round are positions in the words $w_0,\ldots, w_{k-1}$, respectively, that carry the same letter. Write $\Sigma=\{a_0,\ldots,a_{|\Sigma|-1}\}$. We use a register~$\tilde x$ with bound $|\Sigma|-1$. In guess part 0, the action of the
transition from {\em guess$(a_j)$} to {\em end} also sets $\tilde x:=j$. In the guess parts $i<k$ for $i\neq 0$, the guards of the transitions from {\em start} to {\em guess$(a_j)$} check that $\tilde x=j$.

We count rounds using a register $\tilde y$ with  bound $m$. We let the action of the transition from {\em end} in guess part $k-1$ to {\em start} in guess part 0  set $\tilde y:=\tilde y+1$.
From copy 0 of {\em start} there is a transition to  {\em accept} with guard 
$\tilde y= m$. 
This completes the construction of $\A$.

\medskip

To prove the theorem, assume $p$-SWA is fixed-parameter tractable, i.e., there is an algorithm deciding $p$-SWA that on an instance $(\A,w)$ runs in time $f(k')\cdot |w|^{O(1)}$ where~$k'$ is the parameter of the instance, and $f:\N\to\N$ is a nondecreasing computable function. By Theorem~\ref{thm:lcs} is suffices to show that $p$-LCS is fixed-parameter tractable.

Given an instance $(\Sigma, w_0,\ldots,w_{k-1},m)$ of $p$-LCS answer ``no'' if $m>|w_0|$. Otherwise compute the automaton $\A$ as above and then compute an equivalent  stopwatch automaton~$\A'$ as in the construction behind Proposition~\ref{prop:plambda}. It is clear that $(\A',w^m)$ is computable from $(\Sigma, w_0,\ldots,w_{k-1},m)$ in polynomial time (since $m\le|w_0|$).
Then $(\A',w^m)$ is a ``yes''-instance of $p$-SWA if and only if $(\Sigma, w_0,\ldots,w_{k-1},m)$ is a ``yes''-instance of $p$-LCS.  Hence to decide  $p$-LCS it suffices to run the algorithm for $p$-SWA on $(\A',w^m)$. This takes time $f(k')\cdot |w^m|^{O(1)}$ where $k'$ is the parameter of $(\A',w^m)$. By construction, it is clear that $k'\le g(k+|\Sigma|)$ for some computable $g:\N\to\N$ (in fact, $k'\le (k+|\Sigma|)^{O(1)}$). Since $m\le |w_0|$, the  time $f(k')\cdot |w^m|^{O(1)}$ is bounded by $f(g(k+|\Sigma|))\cdot (|w_0|+\cdots +|w_{k-1}|)^{O(1)}$. Thus, $p$-LCS is fixed-parameter tractable.
\end{proof}

Recall, $p$-SWA is meant to formalize the computational problem to be solved by general purpose model-checkers in algorithmic law. Being general purpose, the set of activities $\Sigma$ should be part of the input, it varies with the laws to be modeled. Nevertheless one might ask whether the hardness result in Theorem~\ref{thm:pswa} might be side-stepped by restricting attention to some fixed alphabet $\Sigma$. 

This is unlikely to be the case. Let $p$-SWA($\{0,1\}$) denote the restriction of $p$-SWA to
 instances with $\Sigma=\{0,1\}$. We have the following variant of Theorem~\ref{thm:pswa}:

\begin{theorem} $p$-\textsc{SWA}$(\{0,1\})$ is not fixed-parameter tractable unless $\FPT=\W[1]$.\end{theorem}

\begin{proof}
 Note that the reduction  $(\Sigma, w_0,\ldots,w_{k-1},m)\mapsto (\A',w^m)$ (for $m\le |w_0|$) in the proof above constructs an automaton $\A'$ over the same alphabet $\Sigma$. It is thus a reduction from the restriction of $p$-LCS to instances with $\Sigma=\{0,1\}$ to $p$-\textsc{SWA}$(\{0,1\})$.
  Now, \cite{lcs3} showed that this restriction is $\W[1]$-hard. 
\end{proof}


\section*{Acknowledgements}
We thank Ra\"ul Espejo Boix for a critical reading of Section~\ref{sec:automaton}.
Part of this work has been done while the first author has been employed by Formal Vindications S.L. 
The second author received funding under the following schemes: ICREA Acad\`{e}mia, projects PID2020-115774RB-I00 and PID2019-107667GB-I00 of the Spanish Ministry of Science and Innovation, 2022 DI 051, Generalitat de Catalunya, Departament d’Empresa i Coneixement and 2017 SGR 270 of the AGAUR.
The second author leads the covenant between the University of Barcelona and Formal Vindications S.L.


\providecommand\noopsort[1]{}

\end{document}